\documentclass[preprint,11pt,authoryear,nonatbib]{elsarticle}

\usepackage{amssymb}
\usepackage{amsmath}
\usepackage{amsthm}
\usepackage[natbibapa]{apacite}

\newcommand{\mb}[1]{\mathbf{#1}}
\newcommand{\R}{\mathbb{R}}
\newcommand{\N}{\mathbb{N}}
\newcommand{\p}{\mathbb{P}}
\newcommand{\E}{\mathbb{E}}

\newcommand{\STRATEGY}{\mathsf{Str} }
\newcommand{\success}{\mathsf{Success} }

\newcommand{\block}{\mathsf{Block} }
\newcommand{\bern}{\mathsf{Bernoulli} }
\newcommand{\uniform}{\mathsf{Uniform} }

\usepackage{setspace}
\usepackage{thmtools}
\usepackage[mathscr]{eucal}
\usepackage{bbm}
\usepackage{amsfonts}
\usepackage{xcolor}

\usepackage{amsmath}
\usepackage{amssymb}
\usepackage{graphicx}

\usepackage{algorithm2e}

\newtheorem{theorem}{Theorem}[section]
\newtheorem{lemma}{Lemma}[section]
\newtheorem{corollary}{Corollary}[section]
\newtheorem{definition}{Definition}[section]

\SetKwComment{Comment}{/* }{ */}

\begin{document}

\begin{frontmatter}



\title{Secretary Problems with Random Number
of Candidates: How Prior Distributional Information Helps}

\author[label1]{Junhui Zhang}
\author[label1,label2]{Patrick Jaillet}

\affiliation[label1]{organization={Operations Research Center, MIT},
            country={USA}}

\affiliation[label2]{organization={EECS, MIT},
            country={USA}}



\begin{abstract}
We study variants of the secretary problem, where $N$, the number of candidates, is a random variable with distribution $\mb p$. The decision maker, given prior distributional information about $\mb p$, wants to maximize the probability of success -- picking the largest number among the $N$ candidates -- using only the sequentially revealed relative ranks of the candidates.

When $\mb p$ is fully known (\textit{full information} setting), we show that single-threshold type of strategies can achieve $1/e$-approximation to the optimal probability of success among all possible strategies. {In the sampling setting, we give an algorithm that finds a strategy with additive suboptimality at most $\epsilon$ and provide its sample-complexity guarantee. We also give a sampling lower bound; for small $\epsilon$, the upper and lower bounds match up to a logarithmic factor in $1/\epsilon$.}

\end{abstract}




\end{frontmatter}




\section{Introduction}
Online optimization has gained increasing attention due to its application in a variety of areas such as online revenue management \citep{McGill1999,Elmaghraby2003}, online resource allocation \citep{Agrawal2014,Devanur2019}, online auctions and advertising \citep{Mehta2007_adword,Babaioff2008}, to name a few. In the typical setting, information about a constrained optimization problem is revealed sequentially, and a decision maker must make irrevocable decisions (about some variables of the problem) at each time period, with the final goal to maximize/minimize the objective function without violating the constraints. Different models have been proposed to capture this online nature, from perhaps the purest model of decision making under uncertainty represented by the classical secretary problem \citep{BuchbinderJainSingh2014,Hubert2015,correa2021_sampling,Kleinberg2005}, to more elaborate ones such as general online linear optimization problems \citep{Agrawal2014,Molinaro2012,Kesselheim2014}. However, {most} of these models assume that the number of decisions to be made, i.e., the total number of arrivals, or the time horizon of the problem, is known a priori, and the (successful) algorithms that have been proposed usually require the time horizon as an essential part of their input.

{The secretary problem captures a fundamental aspect of online decision making under a hard budget constraint. We believe that results for this stylized model, especially the block-based learning algorithm developed here, can provide insight into related applications with uncertain horizons, limited selection budgets, and historical samples of arrival counts.}

In many realistic applications, however, the exact size of the sequence of arrivals, or the time horizon, is hardly known a priori. For example, when allocating advertisement slots to online advertisers subject to a daily budget limit, {the decision maker usually does not know} the exact number of advertisers that will show up each day to compete in the successive auctions. Instead, it would be more natural to model it as an unknown, possibly as a random variable with some known or learnable distribution.

In this work, we focus on the secretary problem when the number of candidates, $N$, is a random variable independent of the process, and the decision maker is given prior information about its distribution $\mb p$ in the form of \textit{full information} of $\mb p$ or \textit{i.i.d.\ samples} following $\mb p$. Due to the randomness in $N$, a constant success probability, such as $1/e$, which is attainable in the classical setting, may be impossible, and the optimal strategy may be much more complex than the classical single-threshold type of strategies: that is, reject the first $l$ candidates, then accept the first candidate that is the best among all candidates that have shown up so far \citep{PresmanSonin1973} for some $l\in \N$. Nevertheless, in this paper, we provide two positive messages.

a) Full information setting: $\mb p$ is fully known. We show that single-threshold type of strategies can achieve $1/e$-approximation to the maximum probability of success among all possible strategies.


{b) Sampling setting: i.i.d.\ samples $N^{(1)},\ldots,N^{(m)}\sim \mb p$ are given. We show that $O\left(\frac{1}{\epsilon^2}\max\left\{\log(\frac{1}{\epsilon}),\epsilon \log(\frac{\log(T)}{\epsilon})\right\}\right)$ samples suffice to find a strategy whose success probability is within an additive $\epsilon$ of the optimal success probability when $N\leq T$, with probability at least $1-O(\epsilon)$, and we show a lower bound of $\Omega(\frac{1}{\epsilon^2})$.}


\subsection{Comparison with previous work}
    There have been extensive studies on the secretary problem and its variants, and interested readers can look at \citet{Ferguson1989} and \citet{Samuels92} for a survey. We are particularly interested in the variant where the number of candidates is a random variable, which was first studied by \citet{PresmanSonin1973}. In \citet{PresmanSonin1973}, the authors assume that $\mb p$, the distribution of the time horizon, is known, and derive {an optimal stopping time that depends on $\mb p$}. Later works under a similar setting look at specific distributions: \citet{GNEDIN2004317} assumes $N$ is uniform among $[n]$, \citet{KAWAI20031065} assumes $N$ is uniform and considers choosing either the best or the second best as a success, \citet{gnedin2022lastsuccess} assumes $N$ is negative binomial, to name a few. In a slightly different setting, the secretary problem with random freeze time is considered \citep{SAMUELCAHN1995315,Ester1996}, where $N$ is a random freeze time with upper bound $n$, and the performance of a stopping time $\tau$ (i.e., the strategy is to pick the $\tau$-th arriving candidate) is evaluated using the probability that $\min(\tau,N)$ is the best candidate \textit{among all $n$ candidates}, {rather than only \textit{among the first $N$ candidates}} as in \citet{PresmanSonin1973}.

    \textbf{Comparison with continuous time model.} Another important line of work studies the continuous time model, where {the $i$-th candidate} arrives at time $T_i$, and $T_1,\ldots,T_N\sim F$ are {i.i.d.\ }following a known distribution $F$. With the additional information of the arrival times, the $1/e$-law -- accepting the first best candidate that arrives after time $F^{(-1)}(1/e)$ -- succeeds with probability at least $1/e$ and is asymptotically optimal \citep{Bruss1984Unified,Stewart1981}. Subsequent works extend these results from various perspectives: \citet{BrussSamuels1987} show that similar results hold for general rank-based payoffs, \citet{BRUSS20221059} show that the $1/e$-law is in fact not optimal, \citet{GNEDIN2022226} shows the existence of many strategies that achieve the $1/e$ success probability asymptotically, to name a few. As a comparison, {our model, following \citet{PresmanSonin1973}, is ``harder''} in the sense that any strategy that does not require the arrival times as inputs can be used in the continuous time model, but not the other way around.

\textbf{Full information.} In \citet{PresmanSonin1973}, with full information about the distribution $\mb p$, the problem is reformulated into an optimal stopping problem for a Markov chain whose transition probability and payoff function depend on the distribution $\mb p$ of the time horizon $N$. Classical results from optimal stopping theory then imply an optimal solution of the form $\tau:= \min\{i\mid {\text{the $i$-th number is the largest seen so far}}, i\in \Gamma\}$ for some index set $\Gamma\subset \N$. When $\mb p$ is a delta, uniform, Poisson, or geometric distribution, $\Gamma$ is of the simple form: $\{n^*,n^*+1,\ldots\}$ for some $n^*$ {that depends on} $\E[N]$ and the family of distributions. In addition, the limiting probability of success under $\Gamma$ as $\E[N]\to \infty$ is strictly positive. However, for general $\mb p$, the set $\Gamma$ might be composed of {an infinite number of disjoint intervals}, leading to a rather complicated strategy. We show that the much simpler single-threshold type of strategies, i.e., $\Gamma =\{n^*,n^*+1,\ldots\}$ for some $n^*\in \N$, already achieves a $1/e$-approximation to the optimal, which echoes a recent line of research on efficient approximations to optimal online algorithms \citep{Anari2019,Papadimitriou2021,Braverman2022}.

Another important work on this random time horizon variant is \citet{HamidBatherTrustrum1982}, which studies a particular family of strategies parameterized by $\mb q\in [0,1]^{\N}$, where $q_i$ denotes the probability of selecting the $i$-th candidate given that no candidate has been selected so far and the current candidate is the best among all that have shown up. The focus of \citet{HamidBatherTrustrum1982}, however, is conditions under which $\mb q$ is admissible, i.e., not dominated by another $\mb q'$, for every $\mb p$. Although their goal is different from ours, we adopt their framework and show that any secretary strategy is dominated by some strategy in this family (Lemma \ref{lm:general_algo_to_q}). Thus, in order to maximize the probability of success (given prior information), it suffices to consider strategies in this family.


\textbf{Samples.} When prior information about $N$ comes in the form of samples, the problem resembles a recent line of work on sample-based secretary problems and prophet inequalities \citep{Correa2023,correa2021_sampling,Correa2021b,Correa2019,Azar2014,rubinstein2020}. For instance, in the sample-driven secretary problem in \citet{correa2021_sampling}, each candidate is, independently, included in the sample set with probability $p$. Then the secretary problem restricted to the rest of the candidates is played, with the only difference being the revealed information: at time $t$, the relative ranking of the $t$-th arriving number among the previous $t-1$ numbers \textit{together with the sample set} is revealed. \citet{correa2021_sampling} provides algorithms for each $p\in [0,1)$ {that} succeed with probability $1/e$ when $p=0$ and $0.58$ as $p\to 1$, matching the optimal probability of success in the classical and full information\footnote{In the full information secretary problem, the numbers are {i.i.d.\ }sampled from a known distribution, and each time the value of the number, instead of its relative rank, is revealed.} secretary problem respectively. The setting we consider is different from theirs. We investigate the number of {i.i.d.\ }samples of $N$ needed in order to find a strategy {that} is $\epsilon$-suboptimal as compared to knowing the full distribution of $N$. When $N\leq n$, the naive approach to use the samples to estimate the distribution of $N$ up to accuracy $\epsilon$ in $l_1$ norm (details in Section \ref{sec:sample}) requires $O\left(\frac{n}{\epsilon^2}\right)$ samples in general \citep{Han2015}. In contrast, we provide a block-based algorithm {that} only requires $O\left(\frac{1}{\epsilon^2}\max\left\{\log\left(\frac{1}{\epsilon}\right),\epsilon \log\left(\frac{\log(n)}{\epsilon}\right)\right\}\right)$ samples, and provide a sampling lower bound of $\Omega\left(\frac{1}{\epsilon^2}\right)$.


\subsection{Contribution}
On a high level, the extra randomness in $N$ makes the problem harder than the classical problem. However, due to the inherent random permutation, the distribution $\mb p$ of $N$ affects the performance of strategies only through the sequence $\lambda_i(\mb p):=\sum_{j=i}^{\infty}\frac{1}{j}p_j$, $i\in\N$, interpreted as the (marginal) probability that the $i$-th arriving candidate is the correct one to pick (Lemma \ref{lm:secretary_prob}). In particular, this allows us to improve upon previous works from the following perspectives:
\begin{itemize}
    \item strategy simplicity: in the full information setting, we ``approximate'' the sequence $(\lambda_i(\mb p))_{i\in \N}$ with the sequence $(\lambda_i(\boldsymbol{\delta}^{(n)}))_{i\in \N}$ generated by the delta distribution supported at $n\in \N$ (equivalent to the classical secretary problem with $n$ candidates), and show that single-threshold type of strategies can achieve at least $1/e$ of the optimal probability of success among all strategies;
    \item sampling algorithm: we provide a ``block based'' algorithm {that} can find an $\epsilon$-suboptimal strategy (as compared to knowing the distribution) with high probability, using $O\left(\frac{1}{\epsilon^2}\max\left\{\log\left(\frac{1}{\epsilon}\right),\epsilon \log\left(\frac{\log(T)}{\epsilon}\right)\right\}\right)$ {i.i.d.\ }samples of $N$, where $N\leq T$ with probability $1-O(\epsilon)$. {Specifically, the algorithm partitions the possible horizon values into geometrically growing blocks and moves the probability mass in each block to its right endpoint.} Our algorithm estimates the sequence {$(i\lambda_i(\mb p))_{i\in [T]}$} using this blocked version of the distribution $\mb p$, thereby reducing the number of parameters to be estimated from $T$ to $\log(T)$. Fortunately, due to the random permutation, the performance of any strategy is insensitive to ``local'' changes to the distribution and so the bias caused by the blocked version of the distribution is small. {We also provide a lower bound of $\Omega\left(\frac{1}{\epsilon^2}\right)$; when $\epsilon=O\left(\frac{1}{\log\log(T)}\right)$, the upper and lower bounds match up to a factor of $\log\left(\frac{1}{\epsilon}\right)$.}
\end{itemize}

\subsection{Notations}
We use the following notations. $\lfloor\cdot\rfloor$ and $\lceil\cdot\rceil$ denote the largest integer no greater than and the smallest integer no less than the variable. $\mb 1[\mathcal E]$ denotes the indicator function of the event $\mathcal E${, which is $1$} when $\mathcal E$ happens and $0$ otherwise. $[n] = \{1,2,\ldots,n\}$ for each $n\in \N$. $H_n:=\sum_{i=1}^n1/i$ is the $n$-th Harmonic number. $\Delta:= \{\mb x = (x_1,x_2,\ldots)|\sum_{i=1}^{\infty} x_i = 1,~x_i\geq  0~\forall i\in \N\}$ and $\Delta^{(n)} = \{\mb x\in \Delta|x_i = 0,~\forall i\geq n+1\}$ represent distributions supported on $\N$ and $[n]$ respectively. For $n\in \N$, $\boldsymbol{\delta}^{(n)}\in \Delta^{(n)}$ is the delta distribution supported at $n$, i.e., $\delta_i^{(n)} = \mb 1[i=n]$. $X \sim \mb p$ means $X$ is a random variable sampled from distribution $\mb p$. All the $\log(\cdot)$ have base $e$, i.e., $\log(e)=1$.

\section{Preliminary}\label{sec:preliminary}
In the classical secretary problem, a decision maker wants to maximize the probability of picking the largest number out of $n$ numbers that arrive in a uniform random order. Each time a number arrives, its relative ranking among all numbers that have arrived is revealed, and the decision maker must decide whether to pick it or not immediately and irrevocably. In this work, we consider the variant where, instead of a fixed, known $n$, the total number to be randomly permuted is a random variable $N$ independent of the random permutation, and the decision maker's goal is to maximize the probability of picking the largest number \textit{before the process stops (at time $N$)}. This setup allows us to explore how prior information about its distribution might help in choosing strategies.

In addition, modeling $N$ as a random variable is more realistic than some other common ways to model parameter uncertainty; see Section 1 of the supplementary material for a comparison with an adversarial-horizon formulation.

\subsection{Setup}\label{sec:set-up}
We assume that there is a sequence of numbers $(x_i)_{i\in \N}$ (assuming no tie) and a distribution $\mb p \in \Delta$. First, nature generates $N$ according to the distribution $\mb p$ such that $\p[N=i]=p_i$. On the event $\{N=i\}$, independent of $N$, a permutation $\sigma: [i]\to[i]$ is picked uniformly at random from all permutations on $[i]$. At time $t = 1,2,\ldots,i$, $R_t$, the relative ranking of $x_{\sigma(t)}$ among $\{x_{\sigma(1)},x_{\sigma(2)},\ldots,x_{\sigma(t)}\}$, is revealed. That is, $R_t = |\{i\in [t]|x_{\sigma(i)} \geq x_{\sigma(t)}\}|$. After observing $R_1,R_2,\ldots,R_t$, the decision maker using the strategy $\STRATEGY$  must decide whether to accept or reject $x_{\sigma(t)}$ immediately, and the decision cannot be changed in the future. Due to the random permutation, conditioning on $N=i$, $\{R_t\}_{t\leq i}$ are independent, and $R_t\sim \uniform([t])$.

The goal of the decision maker is to maximize the probability of success, i.e., {picking the candidate with rank $1$ among all $N$ candidates} before the process stops. We denote this event under strategy $\STRATEGY$ as $\success(\STRATEGY) := \cup_{i=1}^{\infty} \success_i(\STRATEGY) $, where $\success_i(\STRATEGY) $ is the following event:
\begin{align*}
    \success_i(\STRATEGY):=&\{N=i\}\cap\{\STRATEGY\text{ picks } x_{\sigma(t)}\text{ for some }t\in [i]\}\\
    &\cap\{x_{\sigma(t) }= \max_{1\leq j\leq i}x_{\sigma(j)}\}.
\end{align*}
With these notations, the classical secretary problem corresponds to the case when the decision maker knows, before picking the strategy, that $\mb p = \boldsymbol{\delta}^{(n)}$ is the delta distribution supported on $\{n\}$, i.e., $\delta^{(n)}_i = \mb 1[n=i]$. By contrast, in this work, we {do not} assume that $\mb p$ is a delta distribution, or is fully known. Instead, the decision maker is given prior information about $\mb p$ in the form of the distribution $\mb p$ itself or {i.i.d.\ }samples under $\mb p$.

To avoid confusion, we distinguish between \textit{strategy} and \textit{algorithm}. By \textit{strategy}, we refer to strategies a decision maker can use for the secretary problem with an unknown time horizon alone. By \textit{algorithm}, we refer to learning algorithms {that}, given prior information about $\mb p$ such as samples, output a strategy $\STRATEGY$ that the decision maker can use.

\subsection{A parameterized family of strategies}
Let $\overline{\mathcal S}$ denote all strategies for the secretary problem with random $N$. That is, conditioning on $N\geq t$, the decision whether to pick the $t$-th number depends only on the first $t$ relative rankings $(R_i)_{i\leq t}\in \prod_{i=1}^t [i]$ and potential randomness due to the strategy itself. However, it makes no sense to pick any number whose relative rank is not 1, and the relative ranks of the first $t-1$ numbers give no information about the \textit{true} rank of the $t$-th number among all $N$ numbers. Thus, without sacrificing performance, we can in fact restrict the strategies to a family of strategies $\STRATEGY(\mb q)$ parameterized by $\mb q\in \mathcal S:= \{({q}'_i)_{i\in \N}|q'_i\in [0,1],~i\in \N\}$, where
$$q_i=\p[\text{ pick } x_{\sigma(i)}|N\geq i,~R_i=1,~\text{no number has been picked}].$$
The strategy $\STRATEGY(\mb q)$ works in the following way: for $t = 1,2,\ldots$, if $N\geq t$, observe the relative rank $R_t$. If $R_t=1$, generate $Y_t\sim \bern(q_t)$ independent of any previously generated $Y_{t'}$. If $Y_t = 1$, accept the $t$-th number (and the process stops); otherwise continue. {In \citet{PresmanSonin1973}, $\Gamma\subset\N$ is a set of arrival indices, and the corresponding strategy selects the first candidate whose relative rank is $1$ and whose arrival index belongs to $\Gamma$.} This corresponds to $\STRATEGY(\mb q)$ where $q_i = 1$ if and only if $i\in \Gamma$.

With a slight abuse of notation, we write $\success(\mb q)$ for
$\success(\STRATEGY(\mb q))$. The proof of the following guarantee is deferred to
\ref{sec:proof-prelim}.
\begin{lemma}\label{lm:general_algo_to_q}
For any strategy $\STRATEGY \in \overline{\mathcal S}$ for the secretary problem with random $N$, there exists $\mb q^{\STRATEGY}\in \mathcal S$, such that for any $\mb p\in \Delta$, and $N\sim \mb p$, $\p[\success(\mb q^{\STRATEGY}) ] \geq \p[\success(\STRATEGY)] $.
\end{lemma}

For convenience, for each $l\in \N$, we let $\mb q^{(l)}\in \mathcal S$, $q_i^{(l)}:=\mb 1[l\leq i]$, which represents the single-threshold strategy that rejects the first $l-1$ numbers, and accepts the first number $t \geq l$ whose relative rank $R_t=1$.

The family of strategies parameterized by $\mb q\in \mathcal S$ is in fact studied in \citet{HamidBatherTrustrum1982}, but the focus there is on conditions under which $\mb q$ is an admissible policy (not dominated by another $\mb q'$ for every $\mb p\in \Delta$). Below, we state its performance, which can be found in \citet{HamidBatherTrustrum1982}, and can be verified by direct calculation.
\begin{lemma}[\citet{HamidBatherTrustrum1982}]\label{lm:secretary_prob}
If $N\sim \mb p\in \Delta$, strategy $\mb q\in \mathcal S$ satisfies
$$\p[\success(\mb q)] = A(\mb p,\mb q) := \sum_{i=1}^{\infty}p_i \left(\frac{1}{i}\sum_{l=1}^{i}U_{l-1}(\mb q)q_l\right) = \sum_{i=1}^{\infty}U_{i-1}(\mb q)q_i\lambda_i(\mb p),$$
where $U_0(\mb q) = 1$, $U_i(\mb q):=\prod_{l=1}^{i}(1-q_l/l)$ is the probability that the strategy has not accepted any number among the first $i$ numbers conditioned on $N\geq i${, and} $\lambda_i(\mb p):= \sum_{l=i}^{\infty}p_l/l$.
\end{lemma}

For any $n\in \N$, \citet{HamidBatherTrustrum1982} constructs a distribution supported on $[n]$, which we denote by $\mb p^{(n)}_{*}$, such that $A(\mb p^{(n)}_{*},\mb q)= 1/H_n(1-U_n(\mb q))\leq 1/H_n$ for any $\mb q\in \mathcal S$, where $H_n = \sum_{i=1}^n 1/i$ is the $n$-th harmonic number:
\begin{equation}\label{eq:p-star}
    p^{(n)}_{*,i} = \frac{1/H_n}{i+1},~i=1,\ldots,n-1,\quad p^{(n)}_{*,n} = 1/H_n, \quad p^{(n)}_{*,i} =0,\quad i\geq n+1.
\end{equation}
Together with Lemma \ref{lm:general_algo_to_q}, we get the following result:
\begin{corollary}\label{cor:performance-p-star}
When $N\sim \mb p^{(n)}_{*}$, as defined in {Equation \eqref{eq:p-star}}, then for any $\STRATEGY\in \overline{\mathcal S}$, $\p[\success(\STRATEGY)]\leq 1/H_n$, and the upper bound is achieved by $\mb q^{(1)}$, i.e., picking the first number.
\end{corollary}



\section{Full information setting}\label{sec:full_main}
Motivated by \citet{PresmanSonin1973}, which shows that for many distributions (such as uniform, Poisson, and geometric) there exist single-threshold strategies {that are optimal}, we show that this is approximately true for general $\mb p$: the optimal single-threshold strategy gives a $1/e$ approximation to the optimal algorithm in $\overline{\mathcal S}$, that is,
$$\sup_{l\in \N}\p[\success(\mb q^{(l)})]\geq \frac{1}{e}\sup_{\STRATEGY\in \overline{\mathcal S}} \p[\success(\STRATEGY)].$$
More precisely, we show the following result.
\begin{theorem}\label{thm:secretary_approx}
{
For any $\mb p\in \Delta$, let $\lambda_i(\mb p):= \sum_{l=i}^{\infty}p_l/l$ for $i\in \N$, and $\theta(\mb p):=\max_{i\in \N}i\lambda_i(\mb p) =K^*(\mb p) \lambda_{K^*(\mb p)}(\mb p) $ for some $K^*(\mb p)\in \N$. Let $s_1^*=1$, $s_2^*=2$, and, for $K\geq 3$, let $s_K^*\in\N$ satisfy $\sum_{i=s_K^*}^{K-1}1/i\leq 1$ and $\sum_{i=s_K^*-1}^{K-1}1/i>1$. Then
$$\theta(\mb p)/e \leq A(\mb p,\mb q^{(s_{K^*(\mb p)}^*)})\leq \sup_{\mb q\in \mathcal S}A(\mb p,\mb q)=\sup_{\STRATEGY\in \overline{\mathcal S}} \p[\success(\STRATEGY)] \leq \theta(\mb p).$$
}
\end{theorem}

The proof of Theorem \ref{thm:secretary_approx} is deferred to \ref{appendix:full}. {Here, we clarify the role of $\theta(\mb p)$. Since $i\lambda_i(\mb p)\leq \theta(\mb p)$ for every $i\in\N$, $\theta(\mb p)$ upper-bounds the success probability of any strategy. On the other hand, choosing $K^*\in\arg\max_{i\in\N}i\lambda_i(\mb p)$ and applying the classical secretary threshold for horizon $K^*$ yields a success probability of at least
\[
\frac{K^*\lambda_{K^*}(\mb p)}{e}=\frac{\theta(\mb p)}{e}.
\]
Thus, the optimal single-threshold strategy achieves at least a $1/e$ fraction of the optimal success probability.}

{The lower and upper bounds in Theorem \ref{thm:secretary_approx} are individually tight with respect to $\theta(\mb p)$. For $\mb p=\boldsymbol{\delta}^{(n)}$, we have $\lambda_i(\boldsymbol{\delta}^{(n)})=\frac{1}{n}\mb 1[i\leq n]$ and $\theta(\boldsymbol{\delta}^{(n)})=1$. As $n\to\infty$, the optimal success probability of the classical secretary problem with $n$ numbers converges to $\frac{1}{e}=\frac{1}{e}\theta(\mb p)$, so the lower bound is asymptotically attained. For the upper bound, the distribution $\mb p^{(n)}_{*}$ defined in Equation \eqref{eq:p-star} satisfies $i\lambda_i(\mb p^{(n)}_{*})=H_n^{-1}\mb 1[i\leq n]$, and hence $\theta(\mb p^{(n)}_{*})=H_n^{-1}$. Corollary \ref{cor:performance-p-star} shows that the optimal success probability is also $H_n^{-1}$, so the upper bound is attained. Thus, these two families of distributions show that neither inequality can be improved individually over the class of all distributions $\mb p$. However, this does not preclude tighter upper or lower bounds in terms of $\theta(\mb p)$ for restricted families of distributions.}

\section{Sample setting}\label{sec:sample}

In this setting, we assume $N\sim \mb p$ for some unknown distribution $\mb p\in \Delta$, and $m$ samples $N^{(1)},\ldots,N^{(m)}\sim_{i.i.d.} \mb p$ are given. We investigate learning algorithms and the sample complexity of finding an $\epsilon$-suboptimal (as compared to knowing $\mb p$) strategy $\hat{\mb q}\in \mathcal S$, such that with high probability,
\begin{equation}\label{eq:sample-obj}
    A(\mb p, \hat{\mb q})\geq \sup_{\mb q\in \mathcal S} A(\mb p,\mb q)-\epsilon.
\end{equation}

We take a learning-theoretic approach: the samples are used to construct an approximation to $A(\mb p, \cdot):\mathcal S \to [0,1]$ viewed as a function of the strategy parameter $\mb q\in \mathcal S$. If the approximation error is at most $\epsilon/2$ for all $\mb q\in \mathcal S$, then any optimal {strategy} for the approximation function satisfies the goal {Inequality} \eqref{eq:sample-obj} (Lemma \ref{lm:error-approximation-guarantee}). Recall that from Lemma \ref{lm:secretary_prob},
$$A(\mb p,\mb q) := \sum_{i=1}^{\infty}p_i \left(\frac{1}{i}\sum_{l=1}^{i}U_{l-1}(\mb q)q_l\right) = \sum_{i=1}^{\infty}U_{i-1}(\mb q)q_i\lambda_i(\mb p),$$
where $U_0(\mb q) = 1$, $U_i(\mb q):=\prod_{l=1}^{i}(1-q_l/l)$ and $\lambda_i(\mb p):= \sum_{l=i}^{\infty}p_l/l$ for $i\in \N$. Since $\frac{1}{i}\sum_{l=1}^{i}U_{l-1}(\mb q)q_l\leq 1 $ for all $\mb q\in \mathcal S$, $A(\cdot,\mb q)$ is $1$-Lipschitz in $\mb p$ with respect to the $l_1$ norm. Thus, estimating $\mb p$ up to accuracy of $\epsilon/2$ in $l_1$ norm, which usually requires $O(\frac{n}{\epsilon^2})$ samples \citep{Han2015} {(assuming $\mb p$ is supported on $[n]$)}, is enough to find an $\epsilon$-suboptimal strategy. However, due to the specific structure of the secretary problem, the function $A(\mb p,\cdot)$ has two nice properties{, which} we state in Section \ref{sec:A-prop}. These properties allow much more efficient sampling algorithms.

We defer the proofs of results in Sections \ref{sec:A-prop} and \ref{sec:algo} to \ref{appendix:sample}, and the results in Section \ref{sec:lower_bd} to the supplementary material.

\subsection{Two properties of $A(\cdot,\cdot)$}\label{sec:A-prop}
First, for any strategy $\mb q\in \mathcal S$, $A(\cdot,\mb q)$ is insensitive to \textit{local} variation in $\mb p$: if some probability mass that belongs to $p_i$ is assigned to $p_{i'}$, as long as $i'$ is close to $i$, the resulting change in the function value is small. To be more precise, notice that by linearity,
$$A(\mb p,\mb q) = \sum_{i=1}^{\infty}p_i\cdot A(\boldsymbol{\delta}^{(i)},\mb q).$$
We show that for any $\mb q \in \mathcal S$, the coefficients of {$p_i$ and $p_{i'}$} are approximately the same if {$i\approx i'$}:
\begin{lemma}\label{lm:robust-local-error}
    {For any $\mb q\in \mathcal S$ and any $i,i'\in \N$, $i\leq i'$,
    $$\frac{i}{i'}A(\boldsymbol{\delta}^{(i)},\mb q)\leq A(\boldsymbol{\delta}^{(i')},\mb q) \leq \frac{i}{i'}A(\boldsymbol{\delta}^{(i)},\mb q) + (1-\frac{i}{i'}).$$}
\end{lemma}
Intuitively, for the classical secretary problem with {$i'$ candidates}, with probability {$\frac{i}{i'}$} the best candidate is among the first {$i$ arrivals}, and the probability of success under this event is {$A(\boldsymbol{\delta}^{(i)},\mb q)$}. For the event that the best candidate {arrives after the $i$-th candidate}, the probability of success is upper bounded by $1$.

Due to this insensitivity, it suffices to consider the ``block'' version for $\mb p$, where the probability mass in each block (interval) is assigned to the end of the block.
\begin{definition}\label{def:block}
For each $\rho>1$, let $I(\rho):=\{ \lceil \rho^l \rceil,~l\in \{0\}\cup \N\} $ be an index set, with elements (in ascending order) $I(\rho)=\{1=I(\rho,1)<I(\rho,2)<\cdots\}$. {Here, $I(\rho,k)$ denotes the $k$-th distinct element of $I(\rho)$; it need not equal $\lceil\rho^{k-1}\rceil$ because the sequence $(\lceil\rho^l\rceil)_{l\geq 0}$ may contain repeated values.} For convenience, denote $I(\rho,0) = 0$. Let $\block:\Delta\times (1,\infty)\to \Delta$ be a mapping such that for $\mb p\in \Delta$, $\rho>1$, $\block(\mb p,\rho)$ is a distribution supported only on $I(\rho)$, where
    $$\block(\mb p,\rho)_i = \begin{cases} 0\quad &i\notin I(\rho)\\
    \sum_{j=I(\rho,l-1)+1}^{I(\rho,l)} p_j\quad & i = I(\rho,l),l\in \N
    \end{cases}.$$
\end{definition}

\begin{corollary}\label{cor:blocking-error}
    For any $\mb p\in \Delta$, any $\mb q\in \mathcal S$ and $\rho>1$, denote $\tilde{\mb p}= \block(\mb p,\rho)$, then
    $$|A(\mb p,\mb q)- A(\tilde{\mb p},\mb q)|\leq \rho - 1.$$
\end{corollary}

The second property we take advantage of is that $A(\cdot,\cdot)$ depends on $\mb p$ only through the sequence $(i\lambda_i(\mb p))_{i\in \N}$, and is $1$-Lipschitz in {$(i\lambda_i(\mb p))_{i\in \N}$} with respect to the $l_{\infty}$ norm. Thus, a good estimation of all the $\lambda_i$'s suffices for a good function approximation:
\begin{lemma}\label{lm:approx-error-prop}
Let $\mb p\in \Delta$ be a fixed distribution, {and let} $\epsilon>0$. Let {$n\in \N$ be} such that for $i\geq n+1$, $i\lambda_i(\mb p) \leq \epsilon$. Let $\mb G\in [0,1]^{n}$ be a sequence of length $n$, and for $i\in [n]$, $|G_i-i\lambda_i(\mb p)|\leq \epsilon$. Then for any $\mb q\in \mathcal S$,
$$|A(\mb p,\mb q) -\sum_{i=1}^{n}U_{i-1}(\mb q)\frac{q_i}{i}\cdot G_i|\leq \epsilon. $$

{Moreover, let $k_1\leq k_2\in\N$ satisfy $p_i=0$ for $i=k_1,\ldots,k_2-1$. If $|G-k_2\lambda_{k_2}(\mb p)|\leq\epsilon$ for some $G\in\R$, then
\[
\left|\frac{i}{k_2}G-i\lambda_i(\mb p)\right|\leq\epsilon
\quad\text{for all }i=k_1,\ldots,k_2-1.
\]}
\end{lemma}
In fact, if the sequence $\mb G$ above is defined as $G_i = i\lambda_i(\hat{\mb p})$ for some distribution $\hat{\mb p}$ whose support is a subset of the support of $\mb p$ (for example, $\hat{\mb p}$ is the empirical distribution of samples from $\mb p$), then the conditions in Lemma \ref{lm:approx-error-prop} can be relaxed to $|G_i-i\lambda_i(\mb p)|\leq \epsilon$ for $i\in [n]$ and $p_i>0$ {by the second part of the lemma}. In other words, the number of parameters that need to be estimated depends only on the number of non-zero terms in the (truncated) distribution, which makes the block-based technique extremely efficient.

\subsection{Block based learning algorithm}\label{sec:algo}

On a high level, our Algorithm \ref{alg:learning_algo} uses samples to construct an approximation to $A(\block(\mb p,\rho),\cdot)$ for $\rho = 1+\epsilon/4$. {Corollary \ref{cor:blocking-error} establishes that the bias due to the use of $\block$, i.e., the difference between $A(\block(\mb p,\rho),\cdot)$ and $A(\mb p,\cdot)$, is small for $\rho$ close to $1$. In addition, Lemma \ref{lm:approx-error-prop} establishes that a good estimation of $i\lambda_i(\block(\mb p,\rho))$ for $i$ in the support of $\block(\mb p,\rho)$ is sufficient for a good function approximation.} Together, we get the following guarantee:
\begin{theorem}\label{thm:sample-ub}
{
Let $0<\epsilon,\delta<1$. Let $T\in \N$ be such that
$\sum_{i=T+1}^{\infty}p_i \leq \epsilon/12$. Suppose either
\[
T=1
\quad\text{and}\quad
m\geq \frac{36}{\epsilon}\log\!\left(\frac{2}{\delta}\right),
\]
or $T\geq2$ and
\[
m\geq \max\left\{
\frac{36}{\epsilon}\log\!\left(\frac{80\log(T)}{\epsilon\delta}\right),
\frac{8}{\epsilon^2}\log\!\left(\frac{2400}{\epsilon^2\delta}\right)
\right\}.
\]
Then, with probability at least $1-\delta$, Algorithm \ref{alg:learning_algo}
outputs $\hat{\mb q}$ satisfying {Inequality} \eqref{eq:sample-obj}.
}
\end{theorem}

\RestyleAlgo{ruled}

\begin{singlespace}
    \begin{algorithm}
        \footnotesize
	\SetAlgoNoLine
	\KwIn{$m$ samples $N^{(1)},N^{(2)},\ldots,N^{(m)}\sim_{i.i.d.} \mb p$, $0<\epsilon\leq 1$. }
	\KwOut{$\hat{\mb q}\in \mathcal S$.}
    $\rho \gets 1+\epsilon/4$\;
    $l_{\max} \gets \inf\{i\in \N,~I(\rho,i)\geq \max_{j\in [m]} N^{(j)}\}$\;
    $\hat{p}_{i} \gets \begin{cases}
        \frac{1}{m}\sum_{j=1}^m \mb 1[I(\rho,l-1)+1\leq N^{(j)}\leq I(\rho,l)]\quad & i=I(\rho,l),~l\in [l_{\max}]\\
        0\quad & o.w.
    \end{cases}$\;
    $G_i \gets i\lambda_{I(\rho,l)}(\hat{\mb p}),\quad I(\rho,l-1)<i\leq I(\rho,l),~l\in [l_{\max}]$\;
    {$\hat{\mb q} \in \arg\max_{\mb q\in \mathcal S}\sum_{i=1}^{I(\rho,l_{\max})}U_{i-1}(\mb q)\frac{q_i}{i}\cdot  G_i$}
	\caption{Block based learning algorithm}
	\label{alg:learning_algo}
\end{algorithm}
\end{singlespace}

We make three additional comments about Algorithm \ref{alg:learning_algo}. First, due to the use of $\block$, Algorithm \ref{alg:learning_algo} only requires a much weaker sampling oracle, where {i.i.d.\ }samples from the less refined distribution $\block(\mb p,\rho)$ are needed. Second, if $T$ is not known in advance, we can sample $N^{(1)},\ldots,N^{(m_1)}\sim_{i.i.d.} \mb p$ where $m_1 = \lceil \frac{12}{\epsilon}\log(\frac{2}{\delta})\rceil$, and then, with probability at least $1-\delta/2$, setting $T= \max_{j\in [m_1]} N^{(j)}$ will satisfy the requirement in Theorem \ref{thm:sample-ub}. Third, in the last step, an optimization problem needs to be solved. Due to its special structure, it can be solved exactly in a backward fashion in $O(I(\rho,l_{\max}))$ steps\footnote{Notice that each $q_i$ only appears in the terms $U_{j-1}(\mb q)\frac{q_j}{j}\cdot G_j$ for $j\geq i$, and once $q_{j}$'s are determined for all $j>i$, the objective function is linear in $q_i$.}. If an $\epsilon'$-suboptimal solution is obtained instead of an exact one, then $\hat{\mb q}$ is at most $(\epsilon+\epsilon')$-suboptimal for the original problem.

{The detailed proof of Theorem \ref{thm:sample-ub} is provided in the supplementary material; below we provide a sketch.} The proof consists of two parts. First, we use concentration bounds to show that with high probability, the tail of $\mb p$ starting from $N_{\max}:=I(\rho,l_{\max})$ is $O(\epsilon)$, and for each $i= I(\rho,l)$ for some $l\in [l_{\max}]$, $G_i$ is a good estimate of $i\lambda_i(\block(\mb p,\rho))$. Precisely, we lower bound the probability of the following good event \eqref{eq:good-event}:
\begin{equation}\tag{$\mathcal E^*$}\label{eq:good-event}
    \sum_{i=N_{\max}+1}^{\infty} p_i \leq \epsilon/12,\text{ and } |G_i - i\lambda_i(\block(\mb p,\rho))|\leq \epsilon/4~i=I(\rho,l),~l\in [l_{\max}].
\end{equation}
There are two technicalities. First, the number of $G_i$'s to be bounded is $l_{\max}$, a random variable{, which} depends on how ``long'' the tail of $\mb p$ is, and which can potentially be much larger than $\log(T)$. For example, if $\mb p$ is such that $p_{T'} = \epsilon/20$ for some $T'>>T$, then if $m=\Omega(\frac{1}{\epsilon})$, with constant probability, $N^{(j)} = T'$ for some $j\in [m]$, and {$l_{\max}=\Omega(\log T')$, which can be much larger than $\log T$}. We fix this by showing that $\frac{1}{m}\sum_{j=1}^{m} \mb 1[N^{(j)}\geq \rho T] = O(\epsilon)$ with high probability, which implies $|G_i|=O(\epsilon)$ for all $i\geq T$ (jointly). Another issue is that if all of {$|G_i-i\lambda_i(\block(\mb p,\rho))|$} in \eqref{eq:good-event} are bounded using Hoeffding's inequality, the sample complexity will be $O(\frac{\log\log(T)}{\epsilon^2})$. {For fixed $\delta$, the refined argument improves this bound to
$O\!\left(\frac{1}{\epsilon^2}\max\{\epsilon\log\log(T),\log(1/\epsilon)\}\right)$.} We upper bound the sum of variances of $G_i$'s, thereby showing that at most $O(\frac{1}{\epsilon^2})$ of them can have large variances, and the rest of the terms can be bounded using the tighter Bernstein's inequality.

For the second part of the proof, we first use {the second part of Lemma \ref{lm:approx-error-prop}} to show that under Event \eqref{eq:good-event}, all {$|G_i-i\lambda_i(\block(\mb p,\rho))|$} are small for all $i\in [N_{\max}]$, which together with Event \eqref{eq:good-event}, implies that the conditions in Lemma \ref{lm:approx-error-prop} hold with $\mb p$ replaced by $\block(\mb p,\rho)$. With the approximation bound in Corollary \ref{cor:blocking-error}, we show that $A(\mb p,\cdot)$ is approximated up to accuracy $\epsilon/2$, and so output $\hat{\mb q}$ is at most $\epsilon$-suboptimal.

\subsection{Sample lower bound}\label{sec:lower_bd}
Below, we show that $\Omega(\frac{1}{\epsilon^{2}})$ samples are in fact necessary.
\begin{theorem}\label{thm:sample-lb}
{
For any $n\geq 3$, there exists a threshold $c_n$ such that for $0< \epsilon< c_n $, at least $m\geq \frac{c_n-\epsilon}{36\epsilon^2}$ samples are needed to find an $\epsilon/3$-suboptimal strategy $\hat{\mb q}$ with probability at least $2/3$. In addition, $\lim_{n\to\infty} c_n=\frac{1}{1+e}$.
}
\end{theorem}

In fact, the hard cases are supported only on $\{1,n\}$, with the probability $p_1 = c_n\pm \epsilon$. Here, $c_n$ is chosen such that the optimal strategy at time $1$, i.e., $q_1$, switches between $0$ and $1$. Thus, to output an $O(\epsilon)$-suboptimal strategy, the algorithm needs to know $p_1$.

As a comparison, {the sample complexity of Algorithm \ref{alg:learning_algo} matches the $\Omega(\frac{1}{\epsilon^2})$ lower bound up to a factor of $\log(\frac{1}{\epsilon})$} in the regime $\epsilon = O(\frac{1}{\log\log(T)})$, and remains sublinear in $T$ for all $\epsilon\in (0,1)$. {It is} still open whether this dimension dependence in the large $\epsilon$ regime is necessary. {However, we conjecture that the sample complexity is at least $\Omega(\log\log T)$. The result below shows that, for $0<\epsilon<0.03$, specifying an $\epsilon$-suboptimal strategy for distributions supported on $[n]$ requires $\Omega(\log\log n)$ bits.}
\begin{theorem}\label{thm:sample-lb-large-eps}
{
For any $0<\epsilon<0.03$ and all sufficiently large $n\in \N$, it requires at least $\Omega(\log\log(n))$ bits to specify an $\epsilon$-suboptimal strategy $\mb q\in \mathcal S$ for distributions supported on $[n]$.
}
\end{theorem}

\section{{Computational study}}\label{sec:computational}

\subsection{Setup}\label{sec:computational-setup}
We evaluate the finite-sample performance of Algorithm \ref{alg:learning_algo}. For support bounds $n=100,1000,10000$, we consider the eight horizon-distribution families and four strategies summarized in Table \ref{tab:computational-setup}. In the table, $k\in[n]$, $r_n=\min\{1/2,4/n\}$, $a_n=A^*(\boldsymbol{\delta}^{(n)})$, $c_n=(a_n-1/n)/(1+a_n-1/n)$, $S_n=\{\operatorname{round}(n^{j/5}):j=0,\ldots,5\}$\footnote{{Here, $\operatorname{round}(x)$ denotes rounding to the nearest integer, with ties rounded to the nearest even integer, as implemented by NumPy's \texttt{rint} function.}}, and proportional formulas are normalized over $[n]$. {The supplementary material provides explicit formulas for the eight distribution families and four strategies, together with optimality-gap and success-rate results for $\epsilon=0.10$, including detailed tables.}

\begin{table}[h]
\centering
\scriptsize
\begin{tabular}{p{0.14\textwidth}p{0.22\textwidth}p{0.54\textwidth}}
\hline
Type & Name & Definition/description \\
\hline
distribution & fixed & $p_k=\mb 1[k=n]$ \\
distribution & uniform & $p_k=1/n$ \\
distribution & geometric & $p_k\propto r_n(1-r_n)^{k-1}$ \\
distribution & Poisson & $p_k\propto (n/4)^{k-1}/(k-1)!$ \\
distribution & bimodal & $p_k=\frac12\mb 1[k=n/10]+\frac12\mb 1[k=n]$ \\
distribution & power law & $p_k\propto k^{-3/2}$ \\
distribution & hard two-point & $p_k=c_n\mb 1[k=1]+(1-c_n)\mb 1[k=n]$ \\
distribution & geometric support & $p_k=|S_n|^{-1}\mb 1[k\in S_n]$ \\
\hline
strategy & block learner & optimize after blocking the empirical distribution \\
strategy & empirical plug-in & optimize directly for the empirical distribution \\
strategy & mean-horizon threshold & classical threshold for the rounded sample mean horizon \\
strategy & median-horizon threshold & classical threshold for the rounded sample median horizon \\
\hline
\end{tabular}
\caption{Distribution families and strategies used in the computational study.}
\label{tab:computational-setup}
\end{table}

For each distribution, we draw $m\in\{1,2,5,10,20,50,100,200,500\}$ i.i.d.\ training samples. The empirical distribution is formed from the training samples, and $s_h^*$ denotes the optimal classical-secretary threshold for a fixed horizon $h$. We use $\epsilon=0.05$ and 500 independent training replications for every $(\mb p,n,m)$ combination. For each learned strategy, we compute $A(\mb p,\hat{\mb q})$ exactly, so there is no test-simulation error. {We report the mean optimality gap $A^*(\mb p)-A(\mb p,\hat{\mb q})$ and the fraction of replications whose gap is at most $\epsilon$.} The experiments use random seed 20260621 and run in a single Python process on an Apple M3 MacBook Air (8 CPU cores and 8 GB memory). {The Python code used to generate all computational results, tables, and figures is provided as supplementary material.}

\subsection{Results}\label{sec:computational-results}
{Figure \ref{fig:computational-005} shows the mean optimality gaps. The computational study yields the following findings.}
\begingroup
\begin{itemize}
    \item \textbf{Useful performance with few samples.} Averaged over all families and support bounds, the block learner's mean gap falls from $0.076$ at $m=1$ to $0.012$ at $m=10$ and $0.003$ at $m=50$.
    \item \textbf{A single estimated horizon can retain substantial error.} The mean- and median-horizon threshold strategies perform well for fixed or concentrated Poisson horizons, but can remain suboptimal for multimodal or heavy-tailed distributions despite more samples. For the bimodal family at $n=10000$ and $m=500$, the mean- and median-horizon gaps are $0.121$ and $0.069$, whereas the block learner's gap is below $0.001$.
    \item \textbf{Blocking and plug-in have similar measured performance.} Their mean-gap difference is below $0.001$ for every tested combination, explaining the nearly coincident blue and orange curves. These experiments do not establish a performance advantage of blocking; its theoretical benefit is the sample-complexity guarantee.
\end{itemize}
\endgroup

\begin{figure}[p]
\centering
\includegraphics[width=\textwidth]{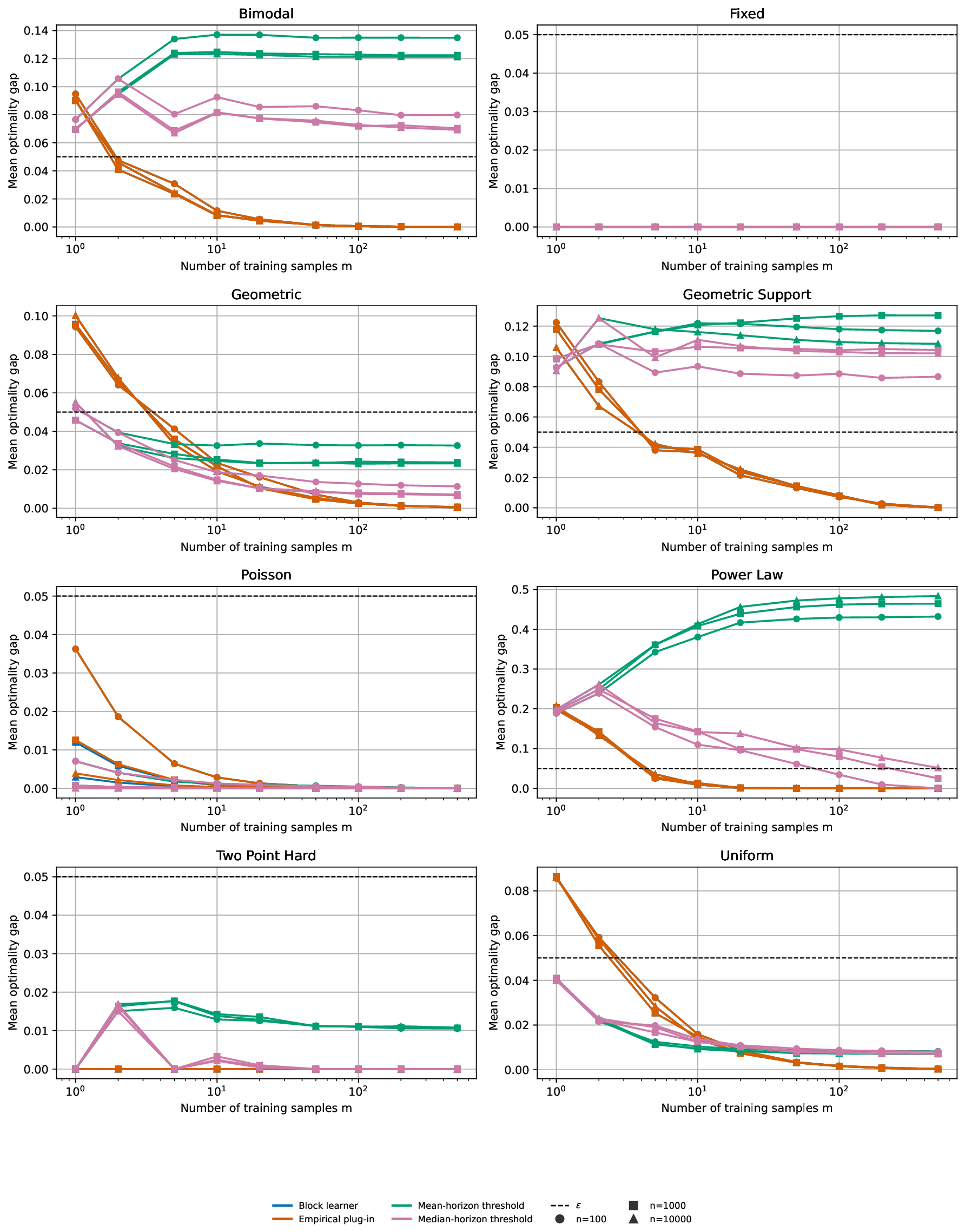}
\caption{{Mean optimality gaps for $\epsilon=0.05$. Colors identify strategies and markers identify support bounds. The blue block-learner and orange empirical plug-in curves nearly coincide.}}
\label{fig:computational-005}
\end{figure}

\clearpage
\section{Conclusion}
In this work, we study variants of the secretary problem with a random number of candidates. In the full information setting, we show that single-threshold policies achieve $1/e$-approximation to the optimal probability of success among all possible strategies. In the sampling setting, we propose a block based learning algorithm. We show that
\[
O\!\left(\frac{1}{\epsilon^2}
\max\left\{\log\!\left(\frac{1}{\epsilon}\right),
\epsilon\log\!\left(\frac{\log(n)}{\epsilon}\right)\right\}\right)
\]
samples suffice to find an $\epsilon$-suboptimal strategy.

{A practically relevant direction for future research is to extend our approach, particularly the block-based approximation and learning techniques, to other online optimization problems with unknown time horizons, such as online resource allocation and online matching.}

\section*{Acknowledgments.}
This work was partly funded by ONR grant N00014-18-1-2122, ONR grant N00014-24-1-2470, AFOSR grant FA9550-19-1-0263, and AFOSR grant FA9550-23-1-0182.

{
\section*{Declaration of generative AI and AI-assisted technologies in the writing process}
During the preparation of this work, the authors used OpenAI Codex for language editing and assistance with \LaTeX{} and Python code. After using this tool, the authors reviewed and edited the content as needed and take full responsibility for the content of the publication.
}

\begin{singlespace}

\appendix

\section{Proof of results in Section \ref{sec:preliminary}}\label{sec:proof-prelim}

\begin{proof}[{Proof of Lemma \ref{lm:general_algo_to_q}}]
Let $\mb p\in \Delta$ denote the distribution for $N$. For any algorithm $\STRATEGY\in \overline{\mathcal S}$, we can consider the following $\mb q\in \mathcal S$:
$$q_i = \p[\STRATEGY \text{ picks }x_{\sigma(i)}|N\geq i, ~R_i = 1,~x_{\sigma(j)} \text{ not picked for }1\leq j\leq i-1].$$
{For this proof, define
\[
U_0(\mb q)=1,\qquad
U_i(\mb q)=\prod_{j=1}^{i}\left(1-\frac{q_j}{j}\right),
\qquad
A(\mb p,\mb q)=\sum_{i=1}^{\infty}U_{i-1}(\mb q)q_i\lambda_i(\mb p).
\]}
For $\STRATEGY$, for any $1\leq i\leq l$, {it is} easy to show by induction (on $i$) that
$$\p[x_{\sigma(j)} \text{ not picked for } 1\leq j\leq i|N\geq l] \leq U_{i}(\mb q).$$
Thus,
$$ \p[\STRATEGY \text{ picks }x_{\sigma(i)}|N\geq i, ~R_i = 1] \leq q_i U_{i-1}(\mb q).$$
And so
\begin{align*}
    &\p[\STRATEGY \text{ picks }x_{\sigma(i)},~x_{\sigma(i)} = \max_{1\leq j\leq N}x_{\sigma(j)}|N\geq i, ~R_i = 1] \\
    =& \p[\STRATEGY \text{ picks }x_{\sigma(i)}|N\geq i, ~R_i = 1]\cdot \p[\arg\max_{1\leq j\leq N}x_{\sigma(j)}\leq i|N\geq i, ~R_i = 1]\\
    \leq&  {q_i U_{i-1}(\mb q) \cdot \sum_{j=i}^{\infty}(\frac{i}{j}\cdot \frac{p_j}{\sum_{j'=i}^{\infty}p_{j'}}) = \frac{i}{\sum_{j'=i}^{\infty}p_{j'}}\cdot q_i U_{i-1}(\mb q) \lambda_i(\mb p).}
\end{align*}
Since $\p[N\geq i, R_i = 1] = \p[ R_i = 1|N\geq i]\cdot \p[N\geq i] = \frac{\sum_{j'=i}^{\infty}p_{j'}}{i}$,
\begin{align*}
    \p[\success(\STRATEGY)] &= \sum_{i=1}^{\infty}\p[N\geq i,~\STRATEGY \text{ picks }x_{\sigma(i)},~x_{\sigma(i)} = \max_{1\leq j\leq N}x_{\sigma(j)}] \\
    &\leq \sum_{i=1}^{\infty}q_iU_{i-1}(\mb q) \lambda_i(\mb p) = A(\mb p,\mb q).\qedhere
\end{align*}
\end{proof}

\section{Proof of results in Section \ref{sec:full_main} }\label{appendix:full}

\begin{proof}[{Proof of Theorem \ref{thm:secretary_approx}}]
For convenience, denote $\theta(\mb p) = \theta$ and $\lambda_i(\mb p) = \lambda_i$ for $i\in \N$. First, notice that $0\leq i\lambda_i\leq \sum_{j=i}^{\infty} p_j \to 0$ as $i\to \infty$, thus by the auxiliary maximum-attainment lemma in the supplementary material, there exists an index $K^*\in \N$ such that
$$\theta = \max_{i\in \N} i\lambda_i = K^*\lambda_{K^*}.$$

Due to {Lemmas} \ref{lm:general_algo_to_q} and \ref{lm:secretary_prob},
$$\sup_{\mb q\in \mathcal S}A(\mb p,\mb q)=\sup_{\STRATEGY\in \overline{\mathcal S}} \p[\success(\STRATEGY)].$$

For the upper bound, since $k\lambda_k\leq \theta$ for all $k$
$$A(\mb p,\mb q)= \sum_{k=1}^{\infty} (\prod_{l=1}^{k-1} (1-q_l/l))(q_k/k)\cdot k\lambda_k\leq \theta \cdot  \sum_{k=1}^{\infty} (\prod_{l=1}^{k-1} (1-q_l/l))(q_k/k).$$
Since for any finite $K\in \N$
$$ \sum_{k=1}^{K} (\prod_{l=1}^{k-1} (1-q_l/l))(q_k/k) = 1-\prod_{l=1}^{K} (1-q_l/l)\leq 1.$$
Taking $K\to \infty$, we get
$$A(\mb p,\mb q) \leq  \theta \cdot\lim_{K\to \infty}  \sum_{k=1}^{K} (\prod_{l=1}^{k-1} (1-q_l/l))(q_k/k)\leq \theta.$$

For the lower bound, the idea is to compare the sequence of $\lambda_k$ with the sequence $\lambda_k^{(T)} = \frac{1}{T} \mb 1[k\leq T]$ generated by a delta distribution at $T$, i.e., $p_i = \mb 1[i=T]$, since we know that for the delta distribution, the problem is just the classical secretary problem with a known time horizon, and achieves a probability of success $\approx 1/e$.

We first truncate the sum at $K\in \N$
$$A(\mb p,\mb q)= \sum_{k=1}^{\infty} (\prod_{l=1}^{k-1} (1-q_l/l))q_k\lambda_k\geq \sum_{k=1}^{K} (\prod_{l=1}^{k-1} (1-q_l/l))q_k\lambda_k.$$
Since $\lambda_k$ is a non-increasing sequence, we have $\lambda_k\geq \lambda_K$ for all $k\leq K$, and so
$$A(\mb p,\mb q)\geq K\lambda_K\cdot \sum_{k=1}^{K} (\prod_{l=1}^{k-1} (1-q_l/l))q_k\cdot \frac{1}{K}.$$

{
For $K = 1$, taking $\mb q=\mb q^{(s_1^*)}=\mb q^{(1)}$, we have $A(\mb p,\mb q)\geq\lambda_1$. For $K = 2$, taking $\mb q=\mb q^{(s_2^*)}=\mb q^{(2)}$, we have $A(\mb p,\mb q)\geq2\lambda_2\cdot \frac{1}{2}\geq 2\lambda_2\cdot \frac{1}{e}$. For $K\geq 3$, take $\mb q=\mb q^{(s_K^*)}$ as defined in the theorem. Then by the auxiliary classical-threshold lemma in the supplementary material,
\begin{align*}
\frac{1}{K}\sum_{k=1}^{K}
\left(\prod_{l=1}^{k-1}(1-q_l/l)\right)q_k
&=\frac{1}{K}\sum_{k=s_K^*}^{K}
\left(\prod_{l=s_K^*}^{k-1}(1-1/l)\right)\\
&=\frac{s_K^*-1}{K}
\left(\sum_{i=s_K^*-1}^{K-1}\frac{1}{i}\right)
\geq\frac{1}{e}.
\end{align*}
which implies that for all $K\in \N$ (and in particular for $K^*$),
\begin{equation*}
A(\mb p,\mb q^{(s_K^*)})\geq K\lambda_K\cdot \frac{1}{e}\implies A(\mb p,\mb q^{(s_{K^*}^*)})\geq  K^*\lambda_{K^*}\cdot \frac{1}{e} =  \frac{1}{e}\cdot \theta.\qedhere
\end{equation*}
}\end{proof}

\section{Proof of results in Section \ref{sec:sample} }\label{appendix:sample}

\begin{lemma}\label{lm:error-approximation-guarantee}
Let $\mb p\in \Delta$ be a fixed distribution, {and let} $\epsilon>0$. Let $f_{\mb p}:\mathcal S\to \R$ be a function such that $|A(\mb p,\mb q)-f_{\mb p}(\mb q)|\leq \epsilon$ for all $\mb q\in \mathcal S$, and assume that there exists $\mb q^*\in \mathcal S$ such that {$f_{\mb p}(\mb q^*)=\sup_{\mb q\in \mathcal S} f_{\mb p}(\mb q)$, then}
$$A(\mb p,\mb q^*)\geq \sup_{\mb q\in \mathcal S}A(\mb p,\mb q)-2\epsilon.$$
\end{lemma}
\begin{proof}[{Proof of Lemma \ref{lm:error-approximation-guarantee}}]
\begin{equation*}
A(\mb p,\mb q^*)\geq f_{\mb p}(\mb q^*)-\epsilon = \sup_{\mb q\in \mathcal S} f_{\mb p}(\mb q)-\epsilon\geq \sup_{\mb q\in \mathcal S} A({\mb p},\mb q) - 2\epsilon .\qedhere
\end{equation*}
\end{proof}

\begin{lemma}\label{lm:ratio-index}
    For any $\rho>1$, we have $I(\rho,l)\leq \rho \cdot (I(\rho,l-1)+1)$ for all $l\in \N$, and $I(\rho,l')\geq \frac{\rho^{l'-l}}{2} \cdot I(\rho,l)$ for all $l'\geq l\geq 1$.
\end{lemma}

\begin{proof}[{Proof of Lemma \ref{lm:ratio-index}}]
Recall that by definition, $I(\rho):=\{ \lceil \rho^l \rceil,~l\in \{0\}\cup \N\} $, and has elements $I(\rho)=\{1=I(\rho,1)<I(\rho,2)<\cdots\}$, with $I(\rho,0) = 0$.

For $l=1$, $I(\rho,l) = 1$ and $I(\rho,l-1)+1 = 1$, the statement holds since $\rho>1$. For any $l\in \N$, $l\geq 2$, let $k\in \{0\}\cup\N$ be the largest integer, such that $ I(\rho,l-1) = \lceil \rho^k\rceil$, then
\begin{align*}
I(\rho,l)
&=\lceil \rho^{k+1}\rceil
\leq \rho^{k+1}+1
\leq \rho^{k+1}+\rho\\
&=\rho(\rho^k+1)
\leq \rho(\lceil\rho^k\rceil+1)
=\rho\bigl(I(\rho,l-1)+1\bigr).
\end{align*}

For the second part, let $I(\rho,l) = \lceil \rho^k\rceil$ for some $k\in \{0\}\cup\N$, then $I(\rho,l')\geq \lceil \rho^{k+l'-l}\rceil$.
\begin{equation*}
 \frac{I(\rho,l)}{I(\rho,l')} \leq \frac{\lceil \rho^{k}\rceil }{\lceil \rho^{k+l'-l}\rceil }\leq \frac{ \rho^{k}+1}{ \rho^{k+l'-l} }\leq \frac{ 2\rho^{k}}{ \rho^{k+l'-l} } = \frac{2}{\rho^{l'-l}}.\qedhere
\end{equation*}
\end{proof}

\begin{proof}[{Proof of Lemma 4.1}]
    For convenience, denote
\[
s_j := A(\boldsymbol{\delta}^{(j)},\mb q)
= \frac{1}{j}\sum_{l=1}^{j}U_{l-1}(\mb q)q_l,
\qquad j\in\N .
\]
For the lower bound, for any $j'\geq j$, we have
$$j's_{j'} = \sum_{l=1}^{j'}U_{l-1}(\mb q)q_l\geq \sum_{l=1}^{j}U_{l-1}(\mb q)q_l= js_j.$$

For the other direction,
$$s_{j'} = \frac{1}{j'}\sum_{l=1}^{j'}U_{l-1}(\mb q)q_l=  \frac{j}{j'}s_j+\frac{1}{j'}\sum_{l=j+1}^{j'}U_{l-1}(\mb q)q_l\leq \frac{j}{j'}s_j+\sum_{l=j+1}^{j'}U_{l-1}(\mb q)\frac{q_l}{l}.$$

Recall that $U_i(\mb q):=\prod_{l=1}^{i}(1-q_l/l)$, the second term above becomes
\begin{align*}
\sum_{l=j+1}^{j'}U_{l-1}(\mb q)\frac{q_l}{l}
&=U_j(\mb q)\left(1-\prod_{l=j+1}^{j'}
\left(1-\frac{q_l}{l}\right)\right)\\
&\leq 1-\frac{j}{j+1}\cdot\frac{j+1}{j+2}
\cdots\frac{j'-1}{j'}
=1-\frac{j}{j'}.
\end{align*}
thus we get
\begin{equation*}
s_{j'}\leq \frac{j}{j'}s_j + (1-\frac{j}{j'}).\qedhere
\end{equation*}
\end{proof}

\begin{proof}[{Proof of Corollary \ref{cor:blocking-error}}]
For convenience, denote $s_i := \frac{1}{i}\sum_{l=1}^{i}U_{l-1}(\mb q)q_l$ and $k_i:=I(\rho,i)$ for $i\in \N$, $k_0=0$. With these notations,
$$A(\mb p,\mb q) = \sum_{i=1}^{\infty}p_i s_i = \sum_{l=1}^{\infty} \sum_{j=k_{l-1}+1}^{k_l} p_js_j.$$

By Lemma \ref{lm:ratio-index}, $\rho(k_{l-1}+1) \geq k_{l}$ for all $l\in \N$, thus by Lemma \ref{lm:robust-local-error} we get
\begin{align*}
A(\mb p,\mb q)
&\leq \rho\sum_{l=1}^{\infty}
\sum_{j=k_{l-1}+1}^{k_l}p_js_{k_l}\\
&=\rho\sum_{l=1}^{\infty}s_{k_l}
\left(\sum_{j=k_{l-1}+1}^{k_l}p_j\right)
=\rho\sum_{l=1}^{\infty}s_{k_l}\tilde p_{k_l}
=\rho A(\tilde{\mb p},\mb q).
\end{align*}

For the upper bound, by Lemma \ref{lm:robust-local-error} again,
\begin{align*}
A(\mb p,\mb q)
&\geq \sum_{l=1}^{\infty}\sum_{j=k_{l-1}+1}^{k_l}
p_j\left(s_{k_l}-\left(1-\frac{1}{\rho}\right)\right)\\
&=\sum_{l=1}^{\infty}\sum_{j=k_{l-1}+1}^{k_l}
p_js_{k_l}-\left(1-\frac{1}{\rho}\right)\\
&=A(\tilde{\mb p},\mb q)-\left(1-\frac{1}{\rho}\right).
\end{align*}

Thus we get the following bounds
\begin{align*}
A(\mb p,\mb q)
&\geq A(\tilde{\mb p},\mb q)-\left(1-\frac{1}{\rho}\right)
\geq A(\tilde{\mb p},\mb q)-(\rho-1),\\
A(\mb p,\mb q)
&\leq \rho A(\tilde{\mb p},\mb q)
\leq A(\tilde{\mb p},\mb q)+(\rho-1).
\end{align*}
which together implies that for all $\mb q$,
\begin{equation*}
|A(\mb p,\mb q)- A(\tilde{\mb p},\mb q)|\leq \rho - 1.\qedhere
\end{equation*}
\end{proof}

\begin{proof}[{Proof of Lemma \ref{lm:approx-error-prop}}]
For any $\mb q\in \mathcal S$, notice that
\begin{equation}\label{eq:sum-U}
    \forall~k\in \N,~\sum_{i=1}^{k}(U_{i-1}(\mb q)q_i/i) = 1-\prod_{i=1}^{k}(1-\frac{q_i}{i})\leq 1\implies \sum_{i=1}^{\infty}(U_{i-1}(\mb q)q_i/i)\leq 1.
\end{equation}
Thus for any $\mb q\in \mathcal S$,
\begin{align*}
    &\quad |A(\mb p,\mb q) -\sum_{i=1}^{n}U_{i-1}(\mb q)\frac{q_i}{i}\cdot G_i|\\&=|\sum_{i=1}^{n}(U_{i-1}(\mb q)q_i/i) \cdot (i\lambda_i(\mb p)-G_i)+ \sum_{i=n+1}^{\infty}(U_{i-1}(\mb q)q_i/i) \cdot i\lambda_i(\mb p)|\\
    &\leq \sum_{i=1}^{n}(U_{i-1}(\mb q)q_i/i) \cdot |i\lambda_i(\mb p)-G_i| + \sum_{i=n+1}^{\infty}(U_{i-1}(\mb q)q_i/i) \cdot |i\lambda_i(\mb p)|\\
    &\leq \epsilon \cdot \sum_{i=1}^{\infty}(U_{i-1}(\mb q)q_i/i)\leq \epsilon.
\end{align*}
{For the second statement, since $p_i=0$ for $i=k_1,\ldots,k_2-1$, we have $\lambda_i(\mb p)=\lambda_{k_2}(\mb p)$ throughout this interval. Therefore,
\begin{equation*}
\left|\frac{i}{k_2}G-i\lambda_i(\mb p)\right|
=\frac{i}{k_2}\left|G-k_2\lambda_{k_2}(\mb p)\right|
\leq\frac{i}{k_2}\epsilon\leq\epsilon.\qedhere
\end{equation*}}
\end{proof}

\begin{proof}[{Proof sketch of Theorem \ref{thm:sample-ub}}]
{The detailed proof is provided in the supplementary material; below we provide a sketch.}

The proof consists of two parts. First, we show that the good event \ref{eq:good-event} holds with high probability. Then, we show that under this event, the output $\hat{\mb q}$ is at most $\epsilon$-suboptimal.

For the first part, write $\overline{\mb p}=\block(\mb p,\rho)$ and $N_{\max}=I(\rho,l_{\max})$. The tail condition in \ref{eq:good-event} follows from the fact that $N_{\max}$ is at least the largest sampled horizon. To control the errors $|G_i-i\lambda_i(\overline{\mb p})|$, we split the block indices into two ranges. For $i\le \rho T$, each $G_i$ is an unbiased estimate of $i\lambda_i(\overline{\mb p})$; a variance-sum bound, together with Hoeffding and Bernstein inequalities, gives the required uniform concentration over the relevant block indices. For $i>\rho T$, the assumption $\sum_{j=T+1}^{\infty}p_j\le \epsilon/12$ and a Chernoff bound imply that both the empirical and blocked tails are small, so the corresponding errors are also at most $\epsilon/4$. Combining these bounds gives $\p[\mathcal E^*]\ge 1-\delta$ under the stated sample-size condition.

For the second part, on the good event, the sequence $\mb G$ uniformly approximates $(i\lambda_i(\overline{\mb p}))_{i\in\mathbb N}$ on all relevant indices. Lemma \ref{lm:approx-error-prop} then implies that the empirical objective optimized in Algorithm \ref{alg:learning_algo} approximates $A(\overline{\mb p},\cdot)$ within $\epsilon/4$. Corollary \ref{cor:blocking-error} gives another $\epsilon/4$ between $A(\overline{\mb p},\cdot)$ and $A(\mb p,\cdot)$. Therefore the objective optimized by Algorithm \ref{alg:learning_algo} is a uniform $\epsilon/2$ approximation to $A(\mb p,\cdot)$, and Lemma \ref{lm:error-approximation-guarantee} implies that the output $\hat{\mb q}$ is at most $\epsilon$-suboptimal.
\end{proof}

 \bibliographystyle{apacite}


 \bibliography{myBib}
\end{singlespace}






\end{document}


\title{Supplementary Material for ``Secretary Problems with Random Number
of Candidates: How Prior Distributional Information Helps''}
\maketitle






\section{Adversarially picked $N$}\label{sec:adv}
In the main paper, we model $N$ as a random variable. This section contrasts that model with an adaptive adversarial-horizon formulation. Consider the setting where $N$ has an upper bound $n$, and $N$ is picked by an adversary who has access to the same ``rolling'' information as the decision maker. That is, at time $t$, the adversary is given the relative ranking of the $t$-th arriving number, as well as the decision maker's decision at time $t-1$, and then needs to decide immediately and irrevocably if $N=t$. We show that the adversary has a strategy such that the decision maker cannot succeed with probability more than $\frac{1}{\sqrt{n}}$. As a comparison, if $N\leq n$ is a random variable, as opposed to being picked by such an adversary, the decision maker has a strategy {that} guarantees a $\Omega(1/{\log(n)})$ probability of success \citep{HamidBatherTrustrum1982,Hill1991_minimax,GharanVondrak2011}.

Assume that $N\leq n$ has an upper bound, and $N$ is picked by an adversary. At time $t$, the adversary is given the relative ranking of the $t$-th arriving number together with the decision maker's decision at time $t-1$, and needs to decide if $N=t$. Consider the following strategy for the adversary: for $t = 1,2,\ldots,\lfloor \sqrt{n}\rfloor$, the adversary chooses $N>t$. At time $t = \lfloor \sqrt{n}\rfloor +1$, if the decision maker has picked a number (among the first $\lfloor \sqrt{n}\rfloor$ arrivals), then the adversary chooses $N=n$, otherwise, the adversary chooses $N = \lfloor \sqrt{n}\rfloor +1$.

\begin{lemma}\label{lm:adv-strategy}
    Under the above strategy of the adversary, for any decision maker, the probability that he picks the largest number among the $N$ numbers that arrive is no more than $\frac{1}{\sqrt{n}}$.
\end{lemma}
\begin{proof}[{Proof of Lemma \ref{lm:adv-strategy}}]
We consider two cases. If the decision maker picks some number among the first $\lfloor \sqrt{n}\rfloor$ arrivals, then the probability that he picks the largest one among $N=n$ numbers is upper-bounded by the probability that the largest one of $n$ numbers is among the first $\lfloor \sqrt{n}\rfloor$ arrivals, which due to the random permutation has probability upper-bounded by $\frac{\lfloor \sqrt{n}\rfloor}{n}\leq \frac{1}{\sqrt{n}}$.

If no number is picked among the first $\lfloor \sqrt{n}\rfloor$ arrivals, if the decision maker does not pick the $\lfloor \sqrt{n}\rfloor+1$-th number, he automatically fails, and if he picks, due to the random permutation, the probability that the $\lfloor \sqrt{n}\rfloor+1$-th number is the largest one among the first $\lfloor \sqrt{n}\rfloor+1$ arrivals is $\frac{1}{\lfloor \sqrt{n}\rfloor+1}\leq \frac{1}{\sqrt{n}}$.

In both cases, the decision maker cannot succeed with a probability more than $\frac{1}{\sqrt{n}}$. \end{proof}

\section{Proof of Theorem 4.1}
\begin{proof}[{Proof of Theorem 4.1}]
The proof consists of two parts. First, we show that the good event $\mathcal E^*$ holds with high probability. Then, we show that under this event, the output $\hat{\mb q}$ is at most $\epsilon$-suboptimal.

For convenience, denote $\overline{\mb p} = \block(\mb p,\rho)$, and $k_i = I(\rho,i)$ for $i\in \N$, $k_0 = 0$, and $N_{\max} = I(\rho, l_{\max})$.

\textbf{Step 1:} Recall that we say the following event is good (denoted by $\mathcal E^*$),
\begin{equation*}
    \sum_{i=N_{\max}+1}^{\infty} p_i \leq \epsilon/12,\text{ and } |G_i - i\lambda_i(\overline{\mb p})|\leq \epsilon/4~i=I(\rho,l),~l\in [l_{\max}].
\end{equation*}

For the first part of the event, since $N_{\max}\geq \max_{j\in [m]} N^{(j)}$,
\begin{align*}
     \p[\sum_{i=N_{\max}+1}^{\infty} p_i > \epsilon/12]
    &\leq \p\left[\sum_{i=\max_{j\in [m]}N^{(j)}+1}^{\infty}
    p_i>\frac{\epsilon}{12}\right]\\
    &=\p\left[\sum_{i=N^{(j)}+1}^{\infty}p_i>\frac{\epsilon}{12}
    \text{ for every }j\in[m]\right]\\
    &\leq (1-\epsilon/12)^m\leq e^{-\epsilon m/12}.
\end{align*}

For the second part, we consider two cases: $i\leq \rho T$, {for which} we apply the standard Hoeffding or Bernstein concentration bound, and $i> \rho T$ (if $N_{\max}> \rho T$), {which we handle} by bounding the tail of $\hat{\mb p}$.

For $i\leq  \rho T$: first, notice that we can rewrite $\hat{\mb p}$ in the following way:
$$\hat{\mb p} = \frac{1}{m}\sum_{j=1}^{m} \sum_{l=1}^{l_{\max}} \mb 1[k_{l-1}+1\leq N^{(j)}\leq k_l]\cdot \boldsymbol{\delta}^{(k_l)}= \frac{1}{m}\sum_{j=1}^{m} \block(\boldsymbol{\delta}^{N^{(j)}},\rho).$$
Then, for any $i = I(\rho,l)$ for some $l\in [l_{\max}]$, by linearity of $\lambda_i$
$$G_i= i\lambda_i(\hat{\mb p}) =\frac{1}{m}\sum_{j=1}^{m} i\lambda_i(\block(\boldsymbol{\delta}^{(N^{(j)})},\rho)).$$
By linearity of the mapping $\block$, if $N^{(j)}\sim \mb p$ for all $j\in [m]$,
\begin{align*}
\E[G_i]
&=\E\!\left[i\lambda_i\!\left(
\block(\boldsymbol{\delta}^{(N)},\rho)\right)\right]\\
&=i\lambda_i\!\left(\block(\E[\boldsymbol{\delta}^{(N)}],\rho)\right)
=i\lambda_i(\block(\mb p,\rho))
=i\lambda_i(\overline{\mb p}).
\end{align*}

Thus, $G_i$ is an unbiased estimator for $i\lambda_i(\overline{\mb p})$ for $i=I(\rho,l)$, $l\in [l_{\max}]$. Next, we show that $G_i$'s are concentrated around their means. To do so, we first notice that for any distribution $\mb p'\in \Delta$ and any $i\in \N$,
$$i\lambda_i(\mb p') = \sum_{j=i}^{\infty} \frac{i}{j}p'_j\leq  \sum_{j=i}^{\infty} p'_j\leq 1.$$
{Since} $G_i = i\lambda_i(\hat{\mb p})$, {we have} $0\leq G_i\leq 1$ for all $i=I(\rho,l)$, $l\in [l_{\max}]$. By Hoeffding's inequality,
{
$$\p[|G_i - i\lambda_i(\overline{\mb p})|> \epsilon/4]\leq 2e^{-m\epsilon^2/8}.$$
}

Next, we will bound $\sum_{l=1}^{l_{\max}} Var[G_i]$, thereby showing that most of the $G_i$'s have small variance, which allows tighter concentration bounds such as Bernstein's inequality.

To be precise, for $i = k_l$ for some $l\in [l_{\max}]$, and $j\in \N$ where $k_{s-1}<j\leq k_{s}$ for some $s\in \N$,
$$i\lambda_i(\block(\boldsymbol{\delta}^{(j)},\rho))=\frac{k_l}{k_s}\mb 1[s\geq l].$$

Since for $N\sim \mb p$, by our definition that $\overline{\mb p} = \block(\mb p,\rho)$, we have
$\p[k_{l'-1}+1 \leq N\leq k_{l'}]=\overline{p}_{k_{l'}}$, thus
$$\E[(i\lambda_i(\block(\boldsymbol{\delta}^{(N)},\rho)))^2] = \sum_{s=l}^{\infty} \overline{p}_{k_s}(\frac{k_l}{k_s})^2\leq \sum_{s=l}^{\infty} \overline{p}_{k_s}\cdot \frac{4}{\rho^{2(s-l)}}. $$
where the last step is by the ratio bound for consecutive block indices. In particular, we get
\begin{align*}
&\E\!\left[\left(k_1\lambda_{k_1}
(\block(\boldsymbol{\delta}^{(N)},\rho))\right)^2\right]\\
&\quad+\left(1-\frac{1}{\rho^2}\right)
\sum_{l=2}^{\infty}\E\!\left[\left(k_l\lambda_{k_l}
(\block(\boldsymbol{\delta}^{(N)},\rho))\right)^2\right]
\leq \sum_{s=1}^{\infty}4\overline p_{k_s}=4.
\end{align*}
which implies that
\begin{align*}
\sum_{l=1}^{l_{\max}}\operatorname{Var}\!\left[
k_l\lambda_{k_l}(\block(\boldsymbol{\delta}^{(N)},\rho))\right]
&\leq \sum_{l=1}^{l_{\max}}\E\!\left[
\left(k_l\lambda_{k_l}
(\block(\boldsymbol{\delta}^{(N)},\rho))\right)^2\right]\\
&\leq \frac{4}{1-1/\rho^2}
\leq \frac{25}{\epsilon}.
\end{align*}

In particular, at most $\min(300/\epsilon^2,l_{\max})$ of them can have variance larger than $\epsilon/12$.
For the rest of these terms, we use Bernstein's inequality
$$\p[|G_i - i\lambda_i(\overline{\mb p})|> \epsilon/4]\leq 2e^{-m\epsilon/8}.$$

Next, we consider the second case where $i>\rho T$. Notice that
$$\sum_{j=1}^m \mb 1[N^{(j)}\geq T+1] \sim Bino(m,\sum_{j=T+1}^{\infty} p_j),$$
Under the assumption, $\sum_{j=T+1}^{\infty} p_j \leq \epsilon/12$, thus by the Chernoff bound, we have
{
$$\p[\frac{1}{m}\sum_{j=1}^m \mb 1[N^{(j)}\geq T+1]  \geq \epsilon/6]\leq e^{-m\epsilon/36}.$$
}

In addition, for any $i =k_l$ for some $l\in [l_{\max}]$, by the triangle inequality,
$$\sum_{j=i}^{\infty} |\hat{p}_j - \overline{p}_j| \leq \sum_{j=i}^{\infty} \hat{p}_j+\sum_{j=i}^{\infty}\overline{p}_j\leq \frac{1}{m}\sum_{j=1}^m \mb 1[N^{(j)}\geq k_{l-1}+1]+\sum_{j=k_{l-1}+1}^{\infty}\overline{p}_j.$$

{
However, for any $k_l>\rho T$, we have $k_{l-1}\geq T$. This is because $k_l>\rho T$ implies $l\geq 2$, in particular $l-1\geq 1$, and so $k_{l-1}\in I(\rho)$. Now let $l_0$ be the largest integer such that $k_{l-1} = \lceil \rho^{l_0}\rceil$, then $k_{l} = \lceil \rho^{l_0+1}\rceil$. By assumption, $ \lceil \rho^{l_0+1}\rceil> \rho T$. If $ \lceil \rho^{l_0}\rceil< T$, then $\rho^{l_0}\leq T-1$, which implies $\lceil \rho^{l_0+1}\rceil\leq \lceil \rho T -\rho\rceil\leq  \rho T -\rho+1 <\rho T$ contradiction. Thus, under the event $\{\frac{1}{m}\sum_{j=1}^m \mb 1[N^{(j)}\geq T+1]<\frac{\epsilon}{6}\}$, which has probability at least $1-e^{-m\epsilon/36}$, for all $i\in I(\rho)$ such that $i>\rho T$,
}
$$\sum_{j=i}^{\infty} |\hat{p}_j - \overline{p}_j| \leq \frac{1}{m}\sum_{j=1}^m \mb 1[N^{(j)}\geq T+1]+\sum_{j=T+1}^{\infty}\overline{p}_j\leq \frac{\epsilon}{6} +\frac{\epsilon}{12}\leq \frac{\epsilon}{4}.$$

Thus
$$|G_i -i\lambda_i(\overline{\mb p})| = |\sum_{j=i}^{\infty} \frac{i}{j}(\hat{p}_j-\overline{p}_j)| \leq \sum_{j=i}^{\infty} \frac{i}{j}|\hat{p}_j-\overline{p}_j|\leq \sum_{j=i}^{\infty} |\hat{p}_j-\overline{p}_j| \leq \epsilon/4. $$

The indices in the first case satisfy $i\leq \rho T$ and $i\in I(\rho)$, i.e., $i = \lceil \rho^l\rceil$ for some $l\in \{0\}\cup\N$. If $T\geq 2$, since $\lceil \rho \rceil = \lceil \rho^2 \rceil = 2$, the number of terms in the first case is upper bounded by
$$1+\frac{\log(T)}{\log(\rho)} \leq 1+\frac{5\log(T)}{\epsilon}\leq \frac{10\log(T)}{\epsilon}.$$
{
where we use $\log(\rho) = \log(1+ \epsilon/4)\geq 0.2\epsilon$ for $0<\epsilon<1$ and $\log(1+x)\geq 0.8x$ for $x\in [0,0.5]$. Thus, if $m\geq \max( \frac{36}{\epsilon}\log(\frac{80\log(T)}{\epsilon\delta}),\frac{8}{\epsilon^2}\log(\frac{2400}{\epsilon^2\delta}))$,
$$\p[\mathcal E^*]\geq 1-e^{-\epsilon m/12}- \frac{600}{\epsilon^2}\cdot e^{-m\epsilon^2/8}-\frac{20\log(T)}{\epsilon}\cdot e^{-\epsilon m/8}-e^{-\epsilon m/36}\geq 1-\delta.$$

If $T=1$, since $\lceil \rho \rceil>\rho\cdot T$, the only term in the first case is $i=1$, also controlled by the event
$\{m^{-1}\sum_{j=1}^m \mathbf{1}[N^{(j)}>1]<\epsilon/6\}$
used in the second case. With $m\geq \frac{36}{\epsilon}\log(\frac{2}{\delta})$
$$\p[\mathcal E^*]\geq 1-e^{-\epsilon m/12}-e^{-\epsilon m/36}\geq 1-\delta.$$
}

\textbf{Step 2:} If the event $\mathcal E^*$ holds, the following series of approximation results hold.
\begin{itemize}
    \item The sequence $\mb G$ approximates the sequence $(i\lambda_i(\overline{\mb p}))_{i\in \N}$ well. Notice that under $\mathcal E^*$, this is already true for $i\in I(\rho)\cap [N_{\max}]$, and here we show that this is true for all $i\in \N$. First, notice that $\hat{\mb p}$ is supported only on $I(\rho)\cap[N_{\max}]$. Thus, for $i\in [N_{\max}]$ but $i\notin I(\rho)$, let $i'$ be the smallest index such that $i'\geq i$ and $i'\in I(\rho)$, then $i'\leq N_{\max}$, and by the second part in $\mathcal E^*$, $|G_{i'} - i'\lambda_{i'}(\overline{\mb p})|\leq \epsilon/4 $. By {the second part of the approximation lemma},
 $$|G_i - i\lambda_i(\overline{\mb p})|\leq \epsilon/4.$$
    \item $\sum_{i=1}^{N_{\max}}U_{i-1}(\mb q)\frac{q_i}{i}\cdot G_i$ approximates $A(\overline{\mb p},\mb q)$ well. This follows by observing that the sequence $\mb G$ satisfies the conditions in the approximation lemma with $n = N_{\max}$, which implies that for any $\mb q\in \mathcal S$,
$$|A(\overline{\mb p},\mb q) -\sum_{i=1}^{N_{\max}}U_{i-1}(\mb q)\frac{q_i}{i}\cdot G_i|\leq \epsilon/4. $$
    \item $A(\overline{\mb p},\mb q)$ approximates $A(\mb p,\mb q)$ well since $\overline{\mb p}$ is the blocked version of $\mb p$: by the blocking-error corollary,
    $$|A(\mb p,\mb q)- A(\overline{\mb p},\mb q)|\leq \rho - 1 = \epsilon/4.$$
\end{itemize}
Given these approximations, by the triangle inequality,
  \begin{align*}
      &\quad |A(\mb p,\mb q) -\sum_{i=1}^{N_{\max}}U_{i-1}(\mb q)\frac{q_i}{i}\cdot G_i|\\
      &\leq |A(\mb p,\mb q)- A(\overline{\mb p},\mb q)| + |A(\overline{\mb p},\mb q) -\sum_{i=1}^{N_{\max}}U_{i-1}(\mb q)\frac{q_i}{i}\cdot G_i|\leq \epsilon/2.
  \end{align*}
By the approximation-to-optimality lemma, $\hat{\mb q}$ is at most $\epsilon$-suboptimal. \end{proof}

\section{Proof of Theorem 4.2}
\begin{proof}[{Proof of Theorem 4.2}]
For convenience, define $A^*:\Delta \to [0,1]$ as
$$A^*(\mb p) = \sup_{\mb q\in \mathcal S} A(\mb p,\mb q).$$
For $n\geq 3$ and $s\in [0,1]$, let $\mb w^{(n,s)}$ be the distribution supported on $\{1,n\}$ with $w^{(n,s)}_1 = s$ and $w^{(n,s)}_n = 1-s$. For any $\mb q\in \mathcal S$, by linearity,
\begin{align*}
A(\mb w^{(n,s)},\mb q)
&=q_1s+(1-s)A(\boldsymbol{\delta}^{(n)},\mb q)\\
&=q_1s+\frac{1-s}{n}\left(
q_1+(1-q_1)\sum_{i=2}^{n}
\prod_{j=2}^{i-1}\left(1-\frac{q_j}{j}\right)q_i\right).
\end{align*}
Let $\mb q^{(n)}_*$ be an optimal strategy for $\boldsymbol{\delta}^{(n)}$, with $A_n:=A(\boldsymbol{\delta}^{(n)},\mb q^{(n)}_*)$. From \citet{GilbertMosteller1966}, for $n\geq 3$, we have $q^{(n)}_{*,1} = 0$, and so
$$\sup_{q_2,\ldots,q_n\in [0,1]}\sum_{i=2}^{n}\prod_{j=2}^{i-1}(1-q_j/j)q_{i}=\sum_{i=2}^{n}\prod_{j=2}^{i-1}(1-q^{(n)}_{*,j}/j)q^{(n)}_{*,i}=nA_n.$$
Thus
$$\sup_{q_2,\ldots,q_n\in [0,1]}A(\mb w^{(n,s)},\mb q) = q_1s+(1-s)(q_1/n+(1-q_1)A_n).$$
To maximize the above expression over $q_1\in [0,1]$, {it is} easy to see that for $s\geq s_n^*:=\frac{A_n-\frac{1}{n}}{1+A_n-\frac{1}{n}}\approx \frac{1}{1+e}$, taking $q_1=1$ is optimal, otherwise $q_1=0$ is optimal. Take $0<\epsilon<\min(s_n^*,1-s_n^*)$, using $A_n\geq 1/e$,
\begin{align*}
A^*(\mb w^{(n,s_n^*+\epsilon)})
&=\left(1-\frac{1}{n}\right)(s_n^*+\epsilon)+\frac{1}{n}\\
&=A^*(\mb w^{(n,s_n^*)})
+\left(1-\frac{1}{n}\right)\epsilon
\geq A^*(\mb w^{(n,s_n^*)})+\epsilon/e,
\end{align*}
$$ A^*(\mb w^{(n,s_n^*-\epsilon)}) =(1-s_n^*+\epsilon)A_n= A^*(\mb w^{(n,s_n^*)}) + \epsilon A_n\geq  A^*(\mb w^{(n,s_n^*)}) + \epsilon/e.$$
Thus for any $\mb q\in \mathcal S$,
\begin{align*}
    &\quad \frac{1}{2}(A(\mb w^{(n,s_n^*+\epsilon)},\mb q) + A(\mb w^{(n,s_n^*-\epsilon)},\mb q))\\
    &= A(\mb w^{(n,s_n^*)},\mb q) \leq A^*(\mb w^{(n,s_n^*)})\\
    &\leq \frac{1}{2}(A^*(\mb w^{(n,s_n^*+\epsilon)})+A^*(\mb w^{(n,s_n^*-\epsilon)}))-\epsilon/e.
\end{align*}

That is, there is no $\mb q\in \mathcal S$, such that at the same time
\begin{align*}
A(\mb w^{(n,s_n^*+\epsilon)},\mb q)
&\geq A^*(\mb w^{(n,s_n^*+\epsilon)})-\epsilon/3,\\
A(\mb w^{(n,s_n^*-\epsilon)},\mb q)
&\geq A^*(\mb w^{(n,s_n^*-\epsilon)})-\epsilon/3.
\end{align*}

As a result, consider $S\sim \mathrm{Uniform}\{s_n^*+\epsilon,s_n^*-\epsilon\}$, and suppose the algorithm is given $m$ i.i.d.\ samples from $\mb w^{(n,S)}$. If an algorithm can find a $\mb q\in \mathcal S$ {that is $\epsilon/3$-suboptimal} for $\mb w^{(n,S)}$, with probability at least $2/3$, using $m$ samples, then this algorithm can correctly predict which value $S$ takes with probability at least $2/3$ using $m$ samples.

For simplicity, for any distribution $\mb p$, use $\otimes^m \mb p$ to denote $\mb p\times \cdots \times \mb p$ the $m$-product distribution of $\mb p$. Then from classical results (textbook reference \citet{wainwright_2019}), the optimal test for $S$ makes mistake with probability $p_e$,
$$p_e = \frac{1}{2}(1-D_{TV}(\otimes^m \mb w^{(n,s_n^*+\epsilon)},\otimes^m \mb w^{(n,s_n^*-\epsilon)})).$$
To achieve $p_e\leq 1/3$, the sample size must satisfy
\[
D_{TV}\!\left(
\otimes^m\mb w^{(n,s_n^*+\epsilon)},
\otimes^m\mb w^{(n,s_n^*-\epsilon)}\right)\geq \frac13,
\]
\begin{align*}
    &\quad 2D^2_{TV}(\otimes^m \mb w^{(n,s_n^*+\epsilon)},\otimes^m \mb w^{(n,s_n^*-\epsilon)})\\
    &\leq D_{KL}(\otimes^m \mb w^{(n,s_n^*+\epsilon)}||\otimes^m \mb w^{(n,s_n^*-\epsilon)}) \quad \text{Pinsker's inequality}\\
    &= mD_{KL}(\mb w^{(n,s_n^*+\epsilon)}||\mb w^{(n,s_n^*-\epsilon)})\quad \text{tensorization for KL-divergence}.
\end{align*}
{
Since $D_{TV}(\mb w^{(n,s_n^*+\epsilon)},\mb w^{(n,s_n^*-\epsilon)})=2\epsilon$, by reverse Pinsker inequality \citep{sason2015reverse}, we have for $\epsilon<s_n^*$,
\begin{align*}
    D_{KL}(\mb w^{(n,s_n^*+\epsilon)}
    \mathbin\Vert\mb w^{(n,s_n^*-\epsilon)})
    &\leq \frac{2D^2_{TV}(
    \mb w^{(n,s_n^*+\epsilon)},
    \mb w^{(n,s_n^*-\epsilon)})}
    {\min(s_n^*-\epsilon,1-s_n^*+\epsilon)}\\
    &=\frac{8\epsilon^2}{\min(s_n^*-\epsilon,1-s_n^*+\epsilon)}.
\end{align*}
However, for any $n\geq 3$, we have $A_n-\frac{1}{n}<A_n<1$, and so $s_n^* =\frac{A_n-\frac{1}{n}}{1+A_n-\frac{1}{n}}<\frac{1}{2} $, simplifying the above to
$$D_{KL}(\mb w^{(n,s_n^*+\epsilon)}||\mb w^{(n,s_n^*-\epsilon)})\leq \frac{8\epsilon^2}{s_n^*-\epsilon}.$$

Thus we have the following bound:
$$m\geq \frac{s_n^*-\epsilon}{36\epsilon^2}.$$
}
Taking
\[
c_n=s_n^*=\frac{A_n-\frac{1}{n}}{1+A_n-\frac{1}{n}},
\]
and using $\lim_{n\to\infty}A_n=1/e$, we obtain
$\lim_{n\to\infty}c_n=1/(e+1)$. \end{proof}
\section{Proof of Theorem 4.3}
\begin{proof}[{Proof of Theorem 4.3}]
{
For convenience, define $A^*:\Delta \to [0,1]$ as
$$A^*(\mb p) = \sup_{\mb q\in \mathcal S} A(\mb p,\mb q).$$
First, notice that for any $k_1,k_2,k_3\in \N$ where for some $\rho>1$, $k_2\geq \rho k_1$, $k_3\geq \rho k_2$, and the mixture $\mb p_{k_1,k_2,k_3} = \frac{1}{3}\sum_{i=1}^3 \boldsymbol{\delta}^{k_i}$, we have
\begin{equation*}
    3\cdot \lambda_i(\mb p_{k_1,k_2,k_3}) = \begin{cases}
        \frac{1}{k_1}+\frac{1}{k_2}+\frac{1}{k_3},&\quad i=1,\ldots,k_1\\
        \frac{1}{k_2}+\frac{1}{k_3},&\quad i=k_1+1,\ldots,k_2\\
        \frac{1}{k_3},&\quad i=k_2+1,\ldots,k_3 \\
        0,&\quad i>k_3\\
    \end{cases}.
\end{equation*}
Thus, {it is} easy to check that
$$3\cdot \theta(\mb p_{k_1,k_2,k_3}) = \max(1+\frac{k_1}{k_2} + \frac{k_1}{k_3},1+\frac{k_2}{k_3}, 1)\leq 1+\frac{2}{\rho}.$$
In particular, by Theorem 3.1, this means that for any strategy $\mb q\in \mathcal S$,
$$\sum_{i=1}^3 A(\boldsymbol{\delta}^{(k_i)},\mb q)= 3A(\mb p_{k_1,k_2,k_3},\mb q)\leq 1+\frac{2}{\rho}.$$
{The bound $1+2/\rho$ is on the sum of the three success probabilities. Thus, for every strategy $\mb q\in\mathcal S$, there is at least one $i\in\{1,2,3\}$ such that
\[
A(\boldsymbol{\delta}^{(k_i)},\mb q)\leq\frac{1}{3}+\frac{2}{3\rho}.
\]
Since $A^*(\boldsymbol{\delta}^{(k_i)})\geq 1/e$, for this index we have
\[
A^*(\boldsymbol{\delta}^{(k_i)})-A(\boldsymbol{\delta}^{(k_i)},\mb q)
\geq\frac{1}{e}-\frac{1}{3}-\frac{2}{3\rho}>\epsilon
\]
whenever $0<\epsilon<1/e-1/3$ and $\rho>(2/3)/(1/e-1/3-\epsilon)$. Consequently, no single strategy is $\epsilon$-suboptimal for all three fixed-horizon distributions.}

Now for any $0<\epsilon<0.03$, take $\rho = \frac{1}{\frac{1}{e}-\frac{1}{3}-\epsilon}<220$. Consider the sequence $(k_i)_{i\in \N}$ where $k_1 = 1$ and $k_{i+1} = \lceil \rho k_i\rceil $ for all $i\in \N$. Then from the argument above, for any triplet $(k_{l_1},k_{l_2},k_{l_3})$, there is no single $\mb q$ {that} is $\epsilon$-suboptimal for all {three} of them. Thus, to specify $\epsilon$-suboptimal strategies for $\{\boldsymbol{\delta}^{(k_i)},i\in [L]\}$, one needs at least $\lceil L/2\rceil $ strategies $\mb q$. This requires $\Omega(\log(L)) $ bits to specify. Taking $n = k_{L} = O(\rho^{L-1})$, we have $L = \Omega(\frac{\log(n)}{\log(\rho)})$, and so the total number of bits required is {$\Omega(\log\log(n))$}.
}\end{proof}

\section{Auxiliary results}
\begin{lemma}[Classical-threshold lower bound]\label{lemma:supp-delta_dist_lb}
 For $T\geq 3$, let $s_T^*\in \N$ be such that $\sum_{i=s_T^*}^{T-1} \frac{1}{i}\leq 1$ but $\sum_{i=s_T^*-1}^{T-1} \frac{1}{i}> 1$, then
 \[
 \frac{s_T^* - 1}{T}\left(\sum_{i=s_T^*-1}^{T-1} \frac{1}{i}\right)\geq 1/e.
 \]
\end{lemma}
\begin{proof}[{Proof of Lemma \ref{lemma:supp-delta_dist_lb}}]
For $T\leq 9$, \citet{GilbertMosteller1966} list the values of $\frac{s_T^* - 1}{T} (\sum_{i=s_T^*-1}^{T-1} \frac{1}{i})$, which are all $\geq 1/e$. Below we assume that $T\geq 10$. Since $\sum_{i=s_T^*-1}^{T-1} \frac{1}{i}\geq 1$, if $\frac{s_T^* - 1}{T}\geq 1/e$, the result holds. Below, we consider the case when $\frac{s_T^* - 1}{T}<1/e$. Next, we give a lower bound for $\frac{s_T^* - 1}{T}$.
$$1\geq \sum_{i={s_T^*}}^{T-1} \frac{1}{i} \geq \int_{s_T^*}^{T} \frac{1}{x} dx  = \log(\frac{T}{s_T^*})\implies \frac{s_T^* - 1}{T} \geq \frac{1}{e}-\frac{1}{T}.$$

Since for any $a\geq 1$, we have
$$\frac{1}{a} \geq \int_{a}^{a+1}\frac{1}{z}dz + \frac{1}{2}\left(\frac{1}{a}-\frac{1}{a+1}\right),$$

we obtain
\begin{align*}
\sum_{i=s_T^*-1}^{T-1}\frac{1}{i}
&\geq \int_{s_T^*-1}^{T}\frac{1}{x}\,dx
+\frac{1}{2}\sum_{i=s_T^*-1}^{T-1}
\left(\frac{1}{i}-\frac{1}{i+1}\right)\\
&=\log\left(\frac{T}{s_T^*-1}\right)
+\frac{1}{2T}\left(\frac{T}{s_T^*-1}-1\right).
\end{align*}
Thus, defining $g_T(x) = -x\log(x)+\frac{1}{2T}(1-x)$, we have
$$\frac{s_T^* - 1}{T} (\sum_{i=s_T^*-1}^{T-1} \frac{1}{i})\geq g_T(\frac{s_T^* - 1}{T}).$$

Since $g_T(\frac{1}{e}) = \frac{1}{e} + (1-\frac{1}{e})/2T>\frac{1}{e}$, and $g_T(x)$ is concave in $x$, as long as $g_T(\frac{1}{e}-\frac{1}{T})\geq \frac{1}{e}$, we have that $g_T(x)\geq \frac{1}{e}$ for all  $x\in [\frac{1}{e}-\frac{1}{T},\frac{1}{e}]$. So it remains to show $g_T(\frac{1}{e}-\frac{1}{T})\geq \frac{1}{e}$. Let's define $h(y) = -(\frac{1}{e}-y)\log(\frac{1}{e}-y) + \frac{1}{2}y(1-\frac{1}{e}+y)$ for $y\in [0,\frac{1}{10}]$, then $g_T(\frac{1}{e}-\frac{1}{T}) = h(\frac{1}{T})$.
Notice that
$$h'(y) = \log(\frac{1}{e}-y)+1+y+\frac{1}{2}\cdot (1-\frac{1}{e}),\quad h''(y) = -\frac{1}{1/e-y}+1\leq 0.$$
And so for $y\in [0,\frac{1}{10}]$,
$$h'(y) \geq h'(\frac{1}{10})= \log(\frac{1}{e}-\frac{1}{10})+1+\frac{1}{10}+\frac{1}{2}\cdot (1-\frac{1}{e})> 0.09>0.$$
Thus $h$ is non-decreasing on $[0,\frac{1}{10}]$, and so for all $y \in [0,\frac{1}{10}]$,
$$ h(y) \geq h(0) = \frac{1}{e}.$$
This ends the proof of the lemma. \end{proof}

\begin{lemma}[Maximum-attainment lemma]\label{lm:supp-max_sup}
    Let $(a_i)_{i\in \N}$ and $(b_i)_{i\in \N}$ be two sequences in $\R$ such that
    \[
    0\leq a_i\leq b_i,\qquad i\in\N.
    \]
    In addition, suppose that $\lim_{n\to\infty} b_n = 0$. Then there exists an index $i^*\in \N$ such that $a_{i^*}\geq a_j$ for all $j\in \N$.
\end{lemma}
\begin{proof}[{Proof of Lemma \ref{lm:supp-max_sup}}]
    If $a_i = 0$ for all $i\in \N$, then taking $i^*=1$ works. Otherwise, let $l\in \N$ be the first index such that $a_l>0$. Since $\lim_{n\to\infty} b_n = 0$, there exists $l'\in \N$ such that for all $n\geq l'$, $a_n\leq b_n\leq a_l$. Thus taking $i^* \in \arg\max_{j\in [l']} a_j$ works. \end{proof}

\begin{lemma}[Hoeffding inequality]\label{lm:supp-hoeffding}
    Let $X_1,\ldots,X_n$ be a sequence of independent random variables such that $0\leq X_i\leq 1$ for all $i\in [n]$ almost surely. Then for all $t>0$,
    \[
    \p\left[\sum_{i=1}^n X_i-\E\left[\sum_{i=1}^n X_i\right]\geq t\right]\leq e^{-2t^2/n}.
    \]
\end{lemma}

\begin{lemma}[Bernstein's inequality]\label{lm:supp-bernstein}
    Let $X_1,\ldots,X_n$ be a sequence of independent random variables. Suppose that there exists $M\in \R$ such that $|X_i-\E[X_i]|\leq M$ for all $i\in [n]$ almost surely. Then for all $t>0$,
    \[
    \p\left[\sum_{i=1}^n X_i-\E\left[\sum_{i=1}^n X_i\right]\geq t\right]\leq e^{-\frac{t^2/2}{V + Mt/3}},\quad V:=\sum_{i=1}^n Var[X_i].
    \]
\end{lemma}

\begin{lemma}[Chernoff bound]\label{lm:supp-chernoff}
    Let $X_1,\ldots,X_n$ be a sequence of independent Bernoulli random variables. Denote $X = \sum_{i=1}^n X_i$ and $\mu = \E[X]$. Then for all $t\geq 0$,
    \[
    \p[X\geq (1+t)\mu]\leq e^{-t^2\mu/(2+t)}.
    \]
    For all $0< t< 1$,
    \[
    \p[X\leq (1-t)\mu]\leq e^{-t^2\mu/2}.
    \]
\end{lemma}

\section{Computational-study setup}
For completeness, we give the exact formulas used for the computational study in Section 5 of the main paper. For each support bound $n\in\{100,1000,10000\}$, let $k\in[n]$, $r_n=\min\{1/2,4/n\}$, $H_{n,3/2}=\sum_{j=1}^n j^{-3/2}$, and $a_n=A^*(\boldsymbol{\delta}^{(n)})$. The eight distribution families are
\begin{align*}
\text{fixed:}\quad
&p_k=\mb 1[k=n],\\
\text{uniform:}\quad
&p_k=1/n,\\
\text{geometric:}\quad
&p_k=\frac{r_n(1-r_n)^{k-1}}{1-(1-r_n)^n},\\
\text{Poisson:}\quad
&p_k=\frac{(n/4)^{k-1}/(k-1)!}
{\sum_{j=1}^n (n/4)^{j-1}/(j-1)!},\\
\text{bimodal:}\quad
&p_k=\tfrac12\mb 1[k=n/10]+\tfrac12\mb 1[k=n],\\
\text{power law:}\quad
&p_k=k^{-3/2}/H_{n,3/2},\\
\text{hard two-point:}\quad
&p_k=c_n\mb 1[k=1]+(1-c_n)\mb 1[k=n],\\[-0.4ex]
&c_n=\frac{a_n-1/n}{1+a_n-1/n},\\
\text{geometric support:}\quad
&p_k=|S_n|^{-1}\mb 1[k\in S_n],\\[-0.4ex]
&S_n=\{\operatorname{round}(n^{j/5}):j=0,\ldots,5\}.
\end{align*}
{Throughout this section, $\operatorname{round}(x)$ denotes rounding to the nearest integer, with ties rounded to the nearest even integer, as implemented by NumPy's \texttt{rint} function.}
The Poisson distribution is that of $1+X$, where $X\sim\operatorname{Poisson}(n/4)$, conditional on $1+X\leq n$. The hard two-point distribution places its mass at the value where the optimal first-period decision switches.

For training samples $N^{(1)},\ldots,N^{(m)}$, define the empirical distribution
\[
\hat{\mb p}^{\,\mathrm{emp}}_k=m^{-1}\sum_{j=1}^m\mb 1[N^{(j)}=k].
\]
The four compared strategies are
\begin{align*}
\text{block:}\quad
&\hat{\mb q}^{\,\mathrm{blk}}\in\arg\max_{\mb q\in\mathcal S}
A\!\left(\block(\hat{\mb p}^{\,\mathrm{emp}},1+\epsilon/4),\mb q\right),\\
\text{plug-in:}\quad
&\hat{\mb q}^{\,\mathrm{emp}}\in\arg\max_{\mb q\in\mathcal S}
A(\hat{\mb p}^{\,\mathrm{emp}},\mb q),\\
\text{mean:}\quad
&\hat{\mb q}^{\,\mathrm{mean}}=\mb q^{(s_{\hat n_{\mathrm{mean}}}^*)},
\quad \hat n_{\mathrm{mean}}=\operatorname{round}\!\left(m^{-1}
\sum_{j=1}^mN^{(j)}\right),\\
\text{median:}\quad
&\hat{\mb q}^{\,\mathrm{med}}=\mb q^{(s_{\hat n_{\mathrm{med}}}^*)},
\quad \hat n_{\mathrm{med}}=\operatorname{round}
\left(\operatorname{median}_{j\in[m]}N^{(j)}\right).
\end{align*}
Here $s_h^*$ is the optimal classical-secretary threshold for a fixed horizon $h$.

\section{Additional computational results for $\epsilon=0.10$}
This section repeats the computational study in Section 5 of the main paper with $\epsilon=0.10$; all other distributions, support bounds, sample sizes, random seed, and replication counts are unchanged. Averaged over all eight families and three values of $n$, the block learner's mean optimality gap is $0.076$, $0.021$, $0.012$, and $0.003$ at $m=1,5,10,50$, respectively. The empirical plug-in again performs almost identically. The simpler horizon-threshold baselines are competitive with one sample but do not consistently improve with more data because a single estimated horizon cannot represent multimodal or heavy-tailed horizon distributions.

\begin{figure}[p]
\centering
\includegraphics[width=\textwidth]{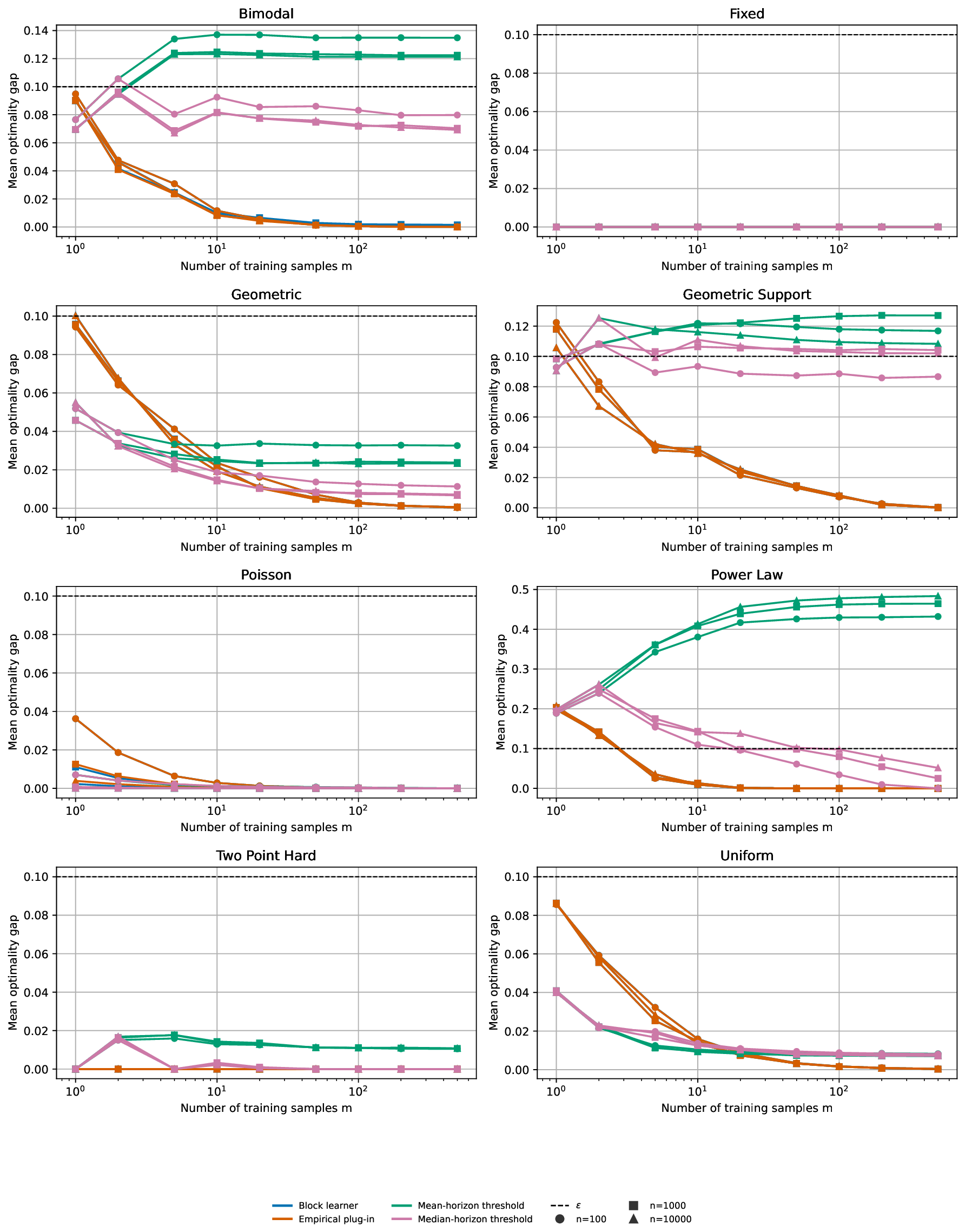}
\caption{{Mean optimality gaps for $\epsilon=0.10$. Colors identify strategies and markers identify support bounds. The blue block-learner and orange empirical plug-in curves nearly coincide.}}
\label{fig:supp-computational-010}
\end{figure}

\clearpage
\begin{landscape}
\begingroup
\scriptsize
\setlength{\tabcolsep}{2pt}
\begin{longtable}{lllrrrrrrrrr}
\caption{Mean additive optimality gap $A^*(\mathbf p)-A(\mathbf p,\widehat{\mathbf q})$ over the training replications.}\label{tab:computational-mean-gap}\\
\toprule
Distribution & $n$ & Method & $m=1$ & $m=2$ & $m=5$ & $m=10$ & $m=20$ & $m=50$ & $m=100$ & $m=200$ & $m=500$ \\
\midrule
\endfirsthead
\toprule
Distribution & $n$ & Method & $m=1$ & $m=2$ & $m=5$ & $m=10$ & $m=20$ & $m=50$ & $m=100$ & $m=200$ & $m=500$ \\
\midrule
\endhead
\midrule
\multicolumn{12}{r}{Continued on next page}\\
\endfoot
\bottomrule
\endlastfoot
Bimodal & 100 & Block learner & 0.095 & 0.048 & 0.031 & 0.012 & 0.006 & 0.001 & 0.001 & 0.000 & 0.000 \\
 & 100 & Empirical plug-in & 0.095 & 0.048 & 0.031 & 0.012 & 0.006 & 0.001 & 0.001 & 0.000 & 0.000 \\
 & 100 & Mean-horizon threshold & 0.077 & 0.106 & 0.134 & 0.137 & 0.137 & 0.135 & 0.135 & 0.135 & 0.135 \\
 & 100 & Median-horizon threshold & 0.077 & 0.106 & 0.080 & 0.093 & 0.085 & 0.086 & 0.083 & 0.080 & 0.080 \\
\addlinespace[2pt]
 & 1000 & Block learner & 0.090 & 0.042 & 0.025 & 0.010 & 0.007 & 0.003 & 0.002 & 0.002 & 0.002 \\
 & 1000 & Empirical plug-in & 0.090 & 0.041 & 0.024 & 0.008 & 0.005 & 0.001 & 0.001 & 0.000 & 0.000 \\
 & 1000 & Mean-horizon threshold & 0.069 & 0.096 & 0.124 & 0.125 & 0.124 & 0.123 & 0.123 & 0.122 & 0.122 \\
 & 1000 & Median-horizon threshold & 0.069 & 0.096 & 0.069 & 0.082 & 0.077 & 0.075 & 0.072 & 0.073 & 0.070 \\
\addlinespace[2pt]
 & 10000 & Block learner & 0.090 & 0.046 & 0.025 & 0.009 & 0.005 & 0.002 & 0.001 & 0.001 & 0.001 \\
 & 10000 & Empirical plug-in & 0.090 & 0.046 & 0.024 & 0.008 & 0.004 & 0.002 & 0.001 & 0.000 & 0.000 \\
 & 10000 & Mean-horizon threshold & 0.070 & 0.095 & 0.123 & 0.123 & 0.123 & 0.121 & 0.121 & 0.121 & 0.121 \\
 & 10000 & Median-horizon threshold & 0.070 & 0.095 & 0.067 & 0.081 & 0.078 & 0.076 & 0.073 & 0.071 & 0.069 \\
\addlinespace[2pt]
\midrule
Fixed & 100 & Block learner & 0.000 & 0.000 & 0.000 & 0.000 & 0.000 & 0.000 & 0.000 & 0.000 & 0.000 \\
 & 100 & Empirical plug-in & 0.000 & 0.000 & 0.000 & 0.000 & 0.000 & 0.000 & 0.000 & 0.000 & 0.000 \\
 & 100 & Mean-horizon threshold & 0.000 & 0.000 & 0.000 & 0.000 & 0.000 & 0.000 & 0.000 & 0.000 & 0.000 \\
 & 100 & Median-horizon threshold & 0.000 & 0.000 & 0.000 & 0.000 & 0.000 & 0.000 & 0.000 & 0.000 & 0.000 \\
\addlinespace[2pt]
 & 1000 & Block learner & 0.000 & 0.000 & 0.000 & 0.000 & 0.000 & 0.000 & 0.000 & 0.000 & 0.000 \\
 & 1000 & Empirical plug-in & 0.000 & 0.000 & 0.000 & 0.000 & 0.000 & 0.000 & 0.000 & 0.000 & 0.000 \\
 & 1000 & Mean-horizon threshold & 0.000 & 0.000 & 0.000 & 0.000 & 0.000 & 0.000 & 0.000 & 0.000 & 0.000 \\
 & 1000 & Median-horizon threshold & 0.000 & 0.000 & 0.000 & 0.000 & 0.000 & 0.000 & 0.000 & 0.000 & 0.000 \\
\addlinespace[2pt]
 & 10000 & Block learner & 0.000 & 0.000 & 0.000 & 0.000 & 0.000 & 0.000 & 0.000 & 0.000 & 0.000 \\
 & 10000 & Empirical plug-in & 0.000 & 0.000 & 0.000 & 0.000 & 0.000 & 0.000 & 0.000 & 0.000 & 0.000 \\
 & 10000 & Mean-horizon threshold & 0.000 & 0.000 & 0.000 & 0.000 & 0.000 & 0.000 & 0.000 & 0.000 & 0.000 \\
 & 10000 & Median-horizon threshold & 0.000 & 0.000 & 0.000 & 0.000 & 0.000 & 0.000 & 0.000 & 0.000 & 0.000 \\
\addlinespace[2pt]
\midrule
Geometric & 100 & Block learner & 0.094 & 0.064 & 0.041 & 0.024 & 0.016 & 0.007 & 0.003 & 0.001 & 0.000 \\
 & 100 & Empirical plug-in & 0.094 & 0.064 & 0.041 & 0.024 & 0.016 & 0.007 & 0.003 & 0.001 & 0.000 \\
 & 100 & Mean-horizon threshold & 0.052 & 0.039 & 0.033 & 0.033 & 0.034 & 0.033 & 0.033 & 0.033 & 0.033 \\
 & 100 & Median-horizon threshold & 0.052 & 0.039 & 0.025 & 0.019 & 0.017 & 0.014 & 0.013 & 0.012 & 0.011 \\
\addlinespace[2pt]
 & 1000 & Block learner & 0.096 & 0.066 & 0.036 & 0.022 & 0.011 & 0.005 & 0.002 & 0.001 & 0.001 \\
 & 1000 & Empirical plug-in & 0.096 & 0.066 & 0.036 & 0.022 & 0.010 & 0.005 & 0.002 & 0.001 & 0.001 \\
 & 1000 & Mean-horizon threshold & 0.046 & 0.034 & 0.028 & 0.025 & 0.024 & 0.023 & 0.024 & 0.024 & 0.024 \\
 & 1000 & Median-horizon threshold & 0.046 & 0.034 & 0.022 & 0.015 & 0.010 & 0.008 & 0.008 & 0.008 & 0.007 \\
\addlinespace[2pt]
 & 10000 & Block learner & 0.100 & 0.068 & 0.033 & 0.019 & 0.011 & 0.005 & 0.003 & 0.001 & 0.001 \\
 & 10000 & Empirical plug-in & 0.100 & 0.068 & 0.033 & 0.019 & 0.011 & 0.005 & 0.003 & 0.001 & 0.001 \\
 & 10000 & Mean-horizon threshold & 0.055 & 0.032 & 0.026 & 0.025 & 0.023 & 0.024 & 0.023 & 0.023 & 0.023 \\
 & 10000 & Median-horizon threshold & 0.055 & 0.032 & 0.020 & 0.014 & 0.010 & 0.009 & 0.007 & 0.007 & 0.007 \\
\addlinespace[2pt]
\midrule
Geometric support & 100 & Block learner & 0.122 & 0.083 & 0.038 & 0.037 & 0.022 & 0.013 & 0.007 & 0.003 & 0.000 \\
 & 100 & Empirical plug-in & 0.122 & 0.083 & 0.038 & 0.037 & 0.022 & 0.013 & 0.007 & 0.003 & 0.000 \\
 & 100 & Mean-horizon threshold & 0.093 & 0.108 & 0.116 & 0.122 & 0.122 & 0.119 & 0.118 & 0.117 & 0.117 \\
 & 100 & Median-horizon threshold & 0.093 & 0.108 & 0.089 & 0.093 & 0.089 & 0.087 & 0.089 & 0.086 & 0.087 \\
\addlinespace[2pt]
 & 1000 & Block learner & 0.118 & 0.078 & 0.040 & 0.039 & 0.024 & 0.015 & 0.008 & 0.002 & 0.000 \\
 & 1000 & Empirical plug-in & 0.118 & 0.078 & 0.040 & 0.039 & 0.024 & 0.014 & 0.008 & 0.002 & 0.000 \\
 & 1000 & Mean-horizon threshold & 0.098 & 0.108 & 0.116 & 0.121 & 0.122 & 0.125 & 0.127 & 0.127 & 0.127 \\
 & 1000 & Median-horizon threshold & 0.098 & 0.108 & 0.103 & 0.106 & 0.106 & 0.105 & 0.104 & 0.105 & 0.104 \\
\addlinespace[2pt]
 & 10000 & Block learner & 0.106 & 0.067 & 0.042 & 0.036 & 0.025 & 0.014 & 0.008 & 0.002 & 0.000 \\
 & 10000 & Empirical plug-in & 0.106 & 0.067 & 0.042 & 0.036 & 0.025 & 0.014 & 0.008 & 0.002 & 0.000 \\
 & 10000 & Mean-horizon threshold & 0.090 & 0.125 & 0.118 & 0.116 & 0.114 & 0.111 & 0.109 & 0.109 & 0.108 \\
 & 10000 & Median-horizon threshold & 0.090 & 0.125 & 0.099 & 0.111 & 0.107 & 0.104 & 0.103 & 0.102 & 0.102 \\
\addlinespace[2pt]
\midrule
Poisson & 100 & Block learner & 0.036 & 0.019 & 0.006 & 0.003 & 0.001 & 0.001 & 0.000 & 0.000 & 0.000 \\
 & 100 & Empirical plug-in & 0.036 & 0.019 & 0.006 & 0.003 & 0.001 & 0.001 & 0.000 & 0.000 & 0.000 \\
 & 100 & Mean-horizon threshold & 0.007 & 0.004 & 0.002 & 0.001 & 0.001 & 0.001 & 0.000 & 0.000 & 0.000 \\
 & 100 & Median-horizon threshold & 0.007 & 0.004 & 0.002 & 0.001 & 0.001 & 0.001 & 0.000 & 0.000 & 0.000 \\
\addlinespace[2pt]
 & 1000 & Block learner & 0.011 & 0.005 & 0.002 & 0.001 & 0.000 & 0.000 & 0.000 & 0.000 & 0.000 \\
 & 1000 & Empirical plug-in & 0.013 & 0.006 & 0.002 & 0.001 & 0.000 & 0.000 & 0.000 & 0.000 & 0.000 \\
 & 1000 & Mean-horizon threshold & 0.001 & 0.000 & 0.000 & 0.000 & 0.000 & 0.000 & 0.000 & 0.000 & 0.000 \\
 & 1000 & Median-horizon threshold & 0.001 & 0.000 & 0.000 & 0.000 & 0.000 & 0.000 & 0.000 & 0.000 & 0.000 \\
\addlinespace[2pt]
 & 10000 & Block learner & 0.002 & 0.001 & 0.000 & 0.000 & 0.000 & 0.000 & 0.000 & 0.000 & 0.000 \\
 & 10000 & Empirical plug-in & 0.004 & 0.002 & 0.001 & 0.000 & 0.000 & 0.000 & 0.000 & 0.000 & 0.000 \\
 & 10000 & Mean-horizon threshold & 0.000 & 0.000 & 0.000 & 0.000 & 0.000 & 0.000 & 0.000 & 0.000 & 0.000 \\
 & 10000 & Median-horizon threshold & 0.000 & 0.000 & 0.000 & 0.000 & 0.000 & 0.000 & 0.000 & 0.000 & 0.000 \\
\addlinespace[2pt]
\midrule
Power law & 100 & Block learner & 0.198 & 0.137 & 0.025 & 0.009 & 0.001 & 0.000 & 0.000 & 0.000 & 0.000 \\
 & 100 & Empirical plug-in & 0.198 & 0.137 & 0.025 & 0.009 & 0.001 & 0.000 & 0.000 & 0.000 & 0.000 \\
 & 100 & Mean-horizon threshold & 0.189 & 0.239 & 0.343 & 0.380 & 0.417 & 0.426 & 0.430 & 0.430 & 0.432 \\
 & 100 & Median-horizon threshold & 0.189 & 0.239 & 0.154 & 0.110 & 0.096 & 0.061 & 0.034 & 0.010 & 0.000 \\
\addlinespace[2pt]
 & 1000 & Block learner & 0.203 & 0.142 & 0.027 & 0.013 & 0.001 & 0.000 & 0.000 & 0.000 & 0.000 \\
 & 1000 & Empirical plug-in & 0.203 & 0.142 & 0.027 & 0.013 & 0.001 & 0.000 & 0.000 & 0.000 & 0.000 \\
 & 1000 & Mean-horizon threshold & 0.194 & 0.249 & 0.361 & 0.408 & 0.439 & 0.456 & 0.462 & 0.464 & 0.464 \\
 & 1000 & Median-horizon threshold & 0.194 & 0.249 & 0.175 & 0.143 & 0.098 & 0.099 & 0.080 & 0.054 & 0.025 \\
\addlinespace[2pt]
 & 10000 & Block learner & 0.207 & 0.133 & 0.036 & 0.010 & 0.001 & 0.000 & 0.000 & 0.000 & 0.000 \\
 & 10000 & Empirical plug-in & 0.207 & 0.133 & 0.036 & 0.010 & 0.001 & 0.000 & 0.000 & 0.000 & 0.000 \\
 & 10000 & Mean-horizon threshold & 0.198 & 0.261 & 0.361 & 0.413 & 0.456 & 0.472 & 0.478 & 0.481 & 0.484 \\
 & 10000 & Median-horizon threshold & 0.198 & 0.261 & 0.164 & 0.142 & 0.138 & 0.102 & 0.098 & 0.077 & 0.052 \\
\addlinespace[2pt]
\midrule
Two-point hard & 100 & Block learner & 0.000 & 0.000 & 0.000 & 0.000 & 0.000 & 0.000 & 0.000 & 0.000 & 0.000 \\
 & 100 & Empirical plug-in & 0.000 & 0.000 & 0.000 & 0.000 & 0.000 & 0.000 & 0.000 & 0.000 & 0.000 \\
 & 100 & Mean-horizon threshold & 0.000 & 0.015 & 0.016 & 0.013 & 0.013 & 0.011 & 0.011 & 0.011 & 0.011 \\
 & 100 & Median-horizon threshold & 0.000 & 0.015 & 0.000 & 0.002 & 0.001 & 0.000 & 0.000 & 0.000 & 0.000 \\
\addlinespace[2pt]
 & 1000 & Block learner & 0.000 & 0.000 & 0.000 & 0.000 & 0.000 & 0.000 & 0.000 & 0.000 & 0.000 \\
 & 1000 & Empirical plug-in & 0.000 & 0.000 & 0.000 & 0.000 & 0.000 & 0.000 & 0.000 & 0.000 & 0.000 \\
 & 1000 & Mean-horizon threshold & 0.000 & 0.016 & 0.018 & 0.014 & 0.014 & 0.011 & 0.011 & 0.011 & 0.011 \\
 & 1000 & Median-horizon threshold & 0.000 & 0.016 & 0.000 & 0.003 & 0.001 & 0.000 & 0.000 & 0.000 & 0.000 \\
\addlinespace[2pt]
 & 10000 & Block learner & 0.000 & 0.000 & 0.000 & 0.000 & 0.000 & 0.000 & 0.000 & 0.000 & 0.000 \\
 & 10000 & Empirical plug-in & 0.000 & 0.000 & 0.000 & 0.000 & 0.000 & 0.000 & 0.000 & 0.000 & 0.000 \\
 & 10000 & Mean-horizon threshold & 0.000 & 0.017 & 0.018 & 0.014 & 0.013 & 0.011 & 0.011 & 0.011 & 0.011 \\
 & 10000 & Median-horizon threshold & 0.000 & 0.017 & 0.000 & 0.002 & 0.000 & 0.000 & 0.000 & 0.000 & 0.000 \\
\addlinespace[2pt]
\midrule
Uniform & 100 & Block learner & 0.086 & 0.059 & 0.032 & 0.016 & 0.009 & 0.003 & 0.002 & 0.001 & 0.000 \\
 & 100 & Empirical plug-in & 0.086 & 0.059 & 0.032 & 0.016 & 0.009 & 0.003 & 0.002 & 0.001 & 0.000 \\
 & 100 & Mean-horizon threshold & 0.040 & 0.022 & 0.012 & 0.010 & 0.009 & 0.009 & 0.008 & 0.008 & 0.008 \\
 & 100 & Median-horizon threshold & 0.040 & 0.022 & 0.020 & 0.014 & 0.011 & 0.009 & 0.009 & 0.008 & 0.008 \\
\addlinespace[2pt]
 & 1000 & Block learner & 0.086 & 0.056 & 0.025 & 0.015 & 0.007 & 0.003 & 0.002 & 0.001 & 0.000 \\
 & 1000 & Empirical plug-in & 0.086 & 0.056 & 0.025 & 0.015 & 0.007 & 0.003 & 0.002 & 0.001 & 0.000 \\
 & 1000 & Mean-horizon threshold & 0.041 & 0.022 & 0.011 & 0.010 & 0.008 & 0.007 & 0.007 & 0.007 & 0.007 \\
 & 1000 & Median-horizon threshold & 0.041 & 0.022 & 0.017 & 0.012 & 0.010 & 0.008 & 0.008 & 0.007 & 0.007 \\
\addlinespace[2pt]
 & 10000 & Block learner & 0.086 & 0.058 & 0.028 & 0.014 & 0.008 & 0.003 & 0.002 & 0.001 & 0.000 \\
 & 10000 & Empirical plug-in & 0.086 & 0.058 & 0.028 & 0.014 & 0.008 & 0.003 & 0.002 & 0.001 & 0.000 \\
 & 10000 & Mean-horizon threshold & 0.040 & 0.023 & 0.012 & 0.009 & 0.008 & 0.008 & 0.007 & 0.007 & 0.007 \\
 & 10000 & Median-horizon threshold & 0.040 & 0.023 & 0.019 & 0.012 & 0.010 & 0.008 & 0.008 & 0.007 & 0.007 \\
\addlinespace[2pt]
\end{longtable}

\endgroup
\end{landscape}

\clearpage
\begin{landscape}
\begingroup
\scriptsize
\setlength{\tabcolsep}{2pt}
\begin{longtable}{lllrrrrrrrrr}
\caption{Fraction of training replications with additive optimality gap at most $\epsilon=0.1$.}\label{tab:computational-epsilon-success}\\
\toprule
Distribution & $n$ & Method & $m=1$ & $m=2$ & $m=5$ & $m=10$ & $m=20$ & $m=50$ & $m=100$ & $m=200$ & $m=500$ \\
\midrule
\endfirsthead
\toprule
Distribution & $n$ & Method & $m=1$ & $m=2$ & $m=5$ & $m=10$ & $m=20$ & $m=50$ & $m=100$ & $m=200$ & $m=500$ \\
\midrule
\endhead
\midrule
\multicolumn{12}{r}{Continued on next page}\\
\endfoot
\bottomrule
\endlastfoot
Bimodal & 100 & Block learner & 0.52 & 0.76 & 0.78 & 0.95 & 0.99 & 1.00 & 1.00 & 1.00 & 1.00 \\
 & 100 & Empirical plug-in & 0.52 & 0.76 & 0.78 & 0.95 & 0.99 & 1.00 & 1.00 & 1.00 & 1.00 \\
 & 100 & Mean-horizon threshold & 0.52 & 0.26 & 0.03 & 0.01 & 0.00 & 0.00 & 0.00 & 0.00 & 0.00 \\
 & 100 & Median-horizon threshold & 0.52 & 0.26 & 0.46 & 0.37 & 0.44 & 0.42 & 0.45 & 0.49 & 0.48 \\
\addlinespace[2pt]
 & 1000 & Block learner & 0.49 & 0.79 & 0.81 & 0.96 & 0.99 & 1.00 & 1.00 & 1.00 & 1.00 \\
 & 1000 & Empirical plug-in & 0.49 & 0.79 & 0.81 & 0.96 & 0.99 & 1.00 & 1.00 & 1.00 & 1.00 \\
 & 1000 & Mean-horizon threshold & 0.49 & 0.25 & 0.03 & 0.01 & 0.00 & 0.00 & 0.00 & 0.00 & 0.00 \\
 & 1000 & Median-horizon threshold & 0.49 & 0.25 & 0.50 & 0.38 & 0.41 & 0.45 & 0.48 & 0.46 & 0.49 \\
\addlinespace[2pt]
 & 10000 & Block learner & 1.00 & 1.00 & 1.00 & 1.00 & 1.00 & 1.00 & 1.00 & 1.00 & 1.00 \\
 & 10000 & Empirical plug-in & 1.00 & 1.00 & 1.00 & 1.00 & 1.00 & 1.00 & 1.00 & 1.00 & 1.00 \\
 & 10000 & Mean-horizon threshold & 1.00 & 0.51 & 0.07 & 0.02 & 0.00 & 0.00 & 0.00 & 0.00 & 0.00 \\
 & 10000 & Median-horizon threshold & 1.00 & 0.51 & 1.00 & 0.74 & 0.82 & 0.87 & 0.93 & 0.94 & 0.96 \\
\addlinespace[2pt]
\midrule
Fixed & 100 & Block learner & 1.00 & 1.00 & 1.00 & 1.00 & 1.00 & 1.00 & 1.00 & 1.00 & 1.00 \\
 & 100 & Empirical plug-in & 1.00 & 1.00 & 1.00 & 1.00 & 1.00 & 1.00 & 1.00 & 1.00 & 1.00 \\
 & 100 & Mean-horizon threshold & 1.00 & 1.00 & 1.00 & 1.00 & 1.00 & 1.00 & 1.00 & 1.00 & 1.00 \\
 & 100 & Median-horizon threshold & 1.00 & 1.00 & 1.00 & 1.00 & 1.00 & 1.00 & 1.00 & 1.00 & 1.00 \\
\addlinespace[2pt]
 & 1000 & Block learner & 1.00 & 1.00 & 1.00 & 1.00 & 1.00 & 1.00 & 1.00 & 1.00 & 1.00 \\
 & 1000 & Empirical plug-in & 1.00 & 1.00 & 1.00 & 1.00 & 1.00 & 1.00 & 1.00 & 1.00 & 1.00 \\
 & 1000 & Mean-horizon threshold & 1.00 & 1.00 & 1.00 & 1.00 & 1.00 & 1.00 & 1.00 & 1.00 & 1.00 \\
 & 1000 & Median-horizon threshold & 1.00 & 1.00 & 1.00 & 1.00 & 1.00 & 1.00 & 1.00 & 1.00 & 1.00 \\
\addlinespace[2pt]
 & 10000 & Block learner & 1.00 & 1.00 & 1.00 & 1.00 & 1.00 & 1.00 & 1.00 & 1.00 & 1.00 \\
 & 10000 & Empirical plug-in & 1.00 & 1.00 & 1.00 & 1.00 & 1.00 & 1.00 & 1.00 & 1.00 & 1.00 \\
 & 10000 & Mean-horizon threshold & 1.00 & 1.00 & 1.00 & 1.00 & 1.00 & 1.00 & 1.00 & 1.00 & 1.00 \\
 & 10000 & Median-horizon threshold & 1.00 & 1.00 & 1.00 & 1.00 & 1.00 & 1.00 & 1.00 & 1.00 & 1.00 \\
\addlinespace[2pt]
\midrule
Geometric & 100 & Block learner & 0.61 & 0.85 & 0.87 & 0.95 & 0.97 & 0.99 & 1.00 & 1.00 & 1.00 \\
 & 100 & Empirical plug-in & 0.61 & 0.85 & 0.87 & 0.95 & 0.97 & 0.99 & 1.00 & 1.00 & 1.00 \\
 & 100 & Mean-horizon threshold & 0.77 & 0.88 & 0.98 & 1.00 & 1.00 & 1.00 & 1.00 & 1.00 & 1.00 \\
 & 100 & Median-horizon threshold & 0.77 & 0.88 & 0.97 & 1.00 & 1.00 & 1.00 & 1.00 & 1.00 & 1.00 \\
\addlinespace[2pt]
 & 1000 & Block learner & 0.65 & 0.85 & 0.96 & 0.99 & 1.00 & 1.00 & 1.00 & 1.00 & 1.00 \\
 & 1000 & Empirical plug-in & 0.65 & 0.85 & 0.97 & 0.99 & 1.00 & 1.00 & 1.00 & 1.00 & 1.00 \\
 & 1000 & Mean-horizon threshold & 0.83 & 0.93 & 0.99 & 1.00 & 1.00 & 1.00 & 1.00 & 1.00 & 1.00 \\
 & 1000 & Median-horizon threshold & 0.83 & 0.93 & 0.99 & 1.00 & 1.00 & 1.00 & 1.00 & 1.00 & 1.00 \\
\addlinespace[2pt]
 & 10000 & Block learner & 0.61 & 0.80 & 0.99 & 1.00 & 1.00 & 1.00 & 1.00 & 1.00 & 1.00 \\
 & 10000 & Empirical plug-in & 0.61 & 0.81 & 0.99 & 1.00 & 1.00 & 1.00 & 1.00 & 1.00 & 1.00 \\
 & 10000 & Mean-horizon threshold & 0.77 & 0.93 & 1.00 & 1.00 & 1.00 & 1.00 & 1.00 & 1.00 & 1.00 \\
 & 10000 & Median-horizon threshold & 0.77 & 0.93 & 0.98 & 1.00 & 1.00 & 1.00 & 1.00 & 1.00 & 1.00 \\
\addlinespace[2pt]
\midrule
Geometric support & 100 & Block learner & 0.34 & 0.55 & 0.86 & 0.93 & 0.97 & 1.00 & 1.00 & 1.00 & 1.00 \\
 & 100 & Empirical plug-in & 0.34 & 0.55 & 0.86 & 0.93 & 0.97 & 1.00 & 1.00 & 1.00 & 1.00 \\
 & 100 & Mean-horizon threshold & 0.52 & 0.41 & 0.21 & 0.06 & 0.03 & 0.00 & 0.00 & 0.00 & 0.00 \\
 & 100 & Median-horizon threshold & 0.52 & 0.41 & 0.53 & 0.56 & 0.59 & 0.54 & 0.52 & 0.54 & 0.52 \\
\addlinespace[2pt]
 & 1000 & Block learner & 0.16 & 0.50 & 0.89 & 0.96 & 0.99 & 1.00 & 1.00 & 1.00 & 1.00 \\
 & 1000 & Empirical plug-in & 0.16 & 0.50 & 0.89 & 0.96 & 0.99 & 1.00 & 1.00 & 1.00 & 1.00 \\
 & 1000 & Mean-horizon threshold & 0.53 & 0.34 & 0.10 & 0.02 & 0.00 & 0.00 & 0.00 & 0.00 & 0.00 \\
 & 1000 & Median-horizon threshold & 0.53 & 0.34 & 0.51 & 0.39 & 0.43 & 0.45 & 0.49 & 0.42 & 0.47 \\
\addlinespace[2pt]
 & 10000 & Block learner & 0.16 & 0.77 & 0.91 & 0.96 & 0.99 & 1.00 & 1.00 & 1.00 & 1.00 \\
 & 10000 & Empirical plug-in & 0.16 & 0.77 & 0.91 & 0.96 & 0.99 & 1.00 & 1.00 & 1.00 & 1.00 \\
 & 10000 & Mean-horizon threshold & 0.50 & 0.07 & 0.19 & 0.05 & 0.01 & 0.00 & 0.00 & 0.00 & 0.00 \\
 & 10000 & Median-horizon threshold & 0.50 & 0.07 & 0.49 & 0.26 & 0.33 & 0.43 & 0.47 & 0.52 & 0.47 \\
\addlinespace[2pt]
\midrule
Poisson & 100 & Block learner & 0.97 & 1.00 & 1.00 & 1.00 & 1.00 & 1.00 & 1.00 & 1.00 & 1.00 \\
 & 100 & Empirical plug-in & 0.97 & 1.00 & 1.00 & 1.00 & 1.00 & 1.00 & 1.00 & 1.00 & 1.00 \\
 & 100 & Mean-horizon threshold & 1.00 & 1.00 & 1.00 & 1.00 & 1.00 & 1.00 & 1.00 & 1.00 & 1.00 \\
 & 100 & Median-horizon threshold & 1.00 & 1.00 & 1.00 & 1.00 & 1.00 & 1.00 & 1.00 & 1.00 & 1.00 \\
\addlinespace[2pt]
 & 1000 & Block learner & 1.00 & 1.00 & 1.00 & 1.00 & 1.00 & 1.00 & 1.00 & 1.00 & 1.00 \\
 & 1000 & Empirical plug-in & 1.00 & 1.00 & 1.00 & 1.00 & 1.00 & 1.00 & 1.00 & 1.00 & 1.00 \\
 & 1000 & Mean-horizon threshold & 1.00 & 1.00 & 1.00 & 1.00 & 1.00 & 1.00 & 1.00 & 1.00 & 1.00 \\
 & 1000 & Median-horizon threshold & 1.00 & 1.00 & 1.00 & 1.00 & 1.00 & 1.00 & 1.00 & 1.00 & 1.00 \\
\addlinespace[2pt]
 & 10000 & Block learner & 1.00 & 1.00 & 1.00 & 1.00 & 1.00 & 1.00 & 1.00 & 1.00 & 1.00 \\
 & 10000 & Empirical plug-in & 1.00 & 1.00 & 1.00 & 1.00 & 1.00 & 1.00 & 1.00 & 1.00 & 1.00 \\
 & 10000 & Mean-horizon threshold & 1.00 & 1.00 & 1.00 & 1.00 & 1.00 & 1.00 & 1.00 & 1.00 & 1.00 \\
 & 10000 & Median-horizon threshold & 1.00 & 1.00 & 1.00 & 1.00 & 1.00 & 1.00 & 1.00 & 1.00 & 1.00 \\
\addlinespace[2pt]
\midrule
Power law & 100 & Block learner & 0.56 & 0.66 & 0.94 & 0.97 & 1.00 & 1.00 & 1.00 & 1.00 & 1.00 \\
 & 100 & Empirical plug-in & 0.56 & 0.66 & 0.94 & 0.97 & 1.00 & 1.00 & 1.00 & 1.00 & 1.00 \\
 & 100 & Mean-horizon threshold & 0.56 & 0.44 & 0.19 & 0.09 & 0.02 & 0.00 & 0.00 & 0.00 & 0.00 \\
 & 100 & Median-horizon threshold & 0.56 & 0.44 & 0.60 & 0.70 & 0.73 & 0.82 & 0.90 & 0.97 & 1.00 \\
\addlinespace[2pt]
 & 1000 & Block learner & 0.53 & 0.64 & 0.93 & 0.96 & 1.00 & 1.00 & 1.00 & 1.00 & 1.00 \\
 & 1000 & Empirical plug-in & 0.53 & 0.64 & 0.93 & 0.96 & 1.00 & 1.00 & 1.00 & 1.00 & 1.00 \\
 & 1000 & Mean-horizon threshold & 0.53 & 0.41 & 0.15 & 0.06 & 0.01 & 0.00 & 0.00 & 0.00 & 0.00 \\
 & 1000 & Median-horizon threshold & 0.53 & 0.41 & 0.54 & 0.60 & 0.71 & 0.70 & 0.75 & 0.83 & 0.92 \\
\addlinespace[2pt]
 & 10000 & Block learner & 0.52 & 0.66 & 0.90 & 0.97 & 1.00 & 1.00 & 1.00 & 1.00 & 1.00 \\
 & 10000 & Empirical plug-in & 0.52 & 0.66 & 0.90 & 0.97 & 1.00 & 1.00 & 1.00 & 1.00 & 1.00 \\
 & 10000 & Mean-horizon threshold & 0.52 & 0.38 & 0.14 & 0.05 & 0.00 & 0.00 & 0.00 & 0.00 & 0.00 \\
 & 10000 & Median-horizon threshold & 0.52 & 0.38 & 0.56 & 0.59 & 0.58 & 0.68 & 0.69 & 0.76 & 0.84 \\
\addlinespace[2pt]
\midrule
Two-point hard & 100 & Block learner & 1.00 & 1.00 & 1.00 & 1.00 & 1.00 & 1.00 & 1.00 & 1.00 & 1.00 \\
 & 100 & Empirical plug-in & 1.00 & 1.00 & 1.00 & 1.00 & 1.00 & 1.00 & 1.00 & 1.00 & 1.00 \\
 & 100 & Mean-horizon threshold & 1.00 & 1.00 & 0.98 & 1.00 & 1.00 & 1.00 & 1.00 & 1.00 & 1.00 \\
 & 100 & Median-horizon threshold & 1.00 & 1.00 & 1.00 & 1.00 & 1.00 & 1.00 & 1.00 & 1.00 & 1.00 \\
\addlinespace[2pt]
 & 1000 & Block learner & 1.00 & 1.00 & 1.00 & 1.00 & 1.00 & 1.00 & 1.00 & 1.00 & 1.00 \\
 & 1000 & Empirical plug-in & 1.00 & 1.00 & 1.00 & 1.00 & 1.00 & 1.00 & 1.00 & 1.00 & 1.00 \\
 & 1000 & Mean-horizon threshold & 1.00 & 1.00 & 0.99 & 1.00 & 1.00 & 1.00 & 1.00 & 1.00 & 1.00 \\
 & 1000 & Median-horizon threshold & 1.00 & 1.00 & 1.00 & 1.00 & 1.00 & 1.00 & 1.00 & 1.00 & 1.00 \\
\addlinespace[2pt]
 & 10000 & Block learner & 1.00 & 1.00 & 1.00 & 1.00 & 1.00 & 1.00 & 1.00 & 1.00 & 1.00 \\
 & 10000 & Empirical plug-in & 1.00 & 1.00 & 1.00 & 1.00 & 1.00 & 1.00 & 1.00 & 1.00 & 1.00 \\
 & 10000 & Mean-horizon threshold & 1.00 & 1.00 & 0.98 & 1.00 & 1.00 & 1.00 & 1.00 & 1.00 & 1.00 \\
 & 10000 & Median-horizon threshold & 1.00 & 1.00 & 1.00 & 1.00 & 1.00 & 1.00 & 1.00 & 1.00 & 1.00 \\
\addlinespace[2pt]
\midrule
Uniform & 100 & Block learner & 0.75 & 0.84 & 0.94 & 0.99 & 1.00 & 1.00 & 1.00 & 1.00 & 1.00 \\
 & 100 & Empirical plug-in & 0.75 & 0.84 & 0.94 & 0.99 & 1.00 & 1.00 & 1.00 & 1.00 & 1.00 \\
 & 100 & Mean-horizon threshold & 0.93 & 0.99 & 1.00 & 1.00 & 1.00 & 1.00 & 1.00 & 1.00 & 1.00 \\
 & 100 & Median-horizon threshold & 0.93 & 0.99 & 1.00 & 1.00 & 1.00 & 1.00 & 1.00 & 1.00 & 1.00 \\
\addlinespace[2pt]
 & 1000 & Block learner & 0.77 & 0.84 & 0.99 & 1.00 & 1.00 & 1.00 & 1.00 & 1.00 & 1.00 \\
 & 1000 & Empirical plug-in & 0.76 & 0.84 & 0.99 & 1.00 & 1.00 & 1.00 & 1.00 & 1.00 & 1.00 \\
 & 1000 & Mean-horizon threshold & 0.94 & 0.99 & 1.00 & 1.00 & 1.00 & 1.00 & 1.00 & 1.00 & 1.00 \\
 & 1000 & Median-horizon threshold & 0.94 & 0.99 & 1.00 & 1.00 & 1.00 & 1.00 & 1.00 & 1.00 & 1.00 \\
\addlinespace[2pt]
 & 10000 & Block learner & 0.75 & 0.85 & 0.98 & 1.00 & 1.00 & 1.00 & 1.00 & 1.00 & 1.00 \\
 & 10000 & Empirical plug-in & 0.75 & 0.84 & 0.98 & 1.00 & 1.00 & 1.00 & 1.00 & 1.00 & 1.00 \\
 & 10000 & Mean-horizon threshold & 0.94 & 0.99 & 1.00 & 1.00 & 1.00 & 1.00 & 1.00 & 1.00 & 1.00 \\
 & 10000 & Median-horizon threshold & 0.94 & 0.99 & 1.00 & 1.00 & 1.00 & 1.00 & 1.00 & 1.00 & 1.00 \\
\addlinespace[2pt]
\end{longtable}

\endgroup
\end{landscape}

 \bibliographystyle{apacite}

 \bibliography{myBib}